\RequirePackage{tikz-cd}
\documentclass[sn-mathphys-num]{sn-jnl}
\usepackage{lineno,hyperref}
\usepackage{url}
\usepackage{stmaryrd}
\usepackage{tikz-cd}
\usepackage{yfonts}
\usepackage[utf8]{inputenc}
\usepackage[T1]{fontenc}
\usepackage [autostyle]{csquotes}
\usepackage{caption}
\usepackage{amsfonts}
\usepackage{amsthm}
\usepackage{eqnarray}
\usepackage{mathtools}

\usepackage{float}
\usepackage{amssymb}
\usepackage{amsmath}
\usepackage{enumerate}
\usepackage{lipsum}
\theoremstyle{thmstyleone}
\newtheorem{thm}{Theorem}
\newtheorem{lem}{Lemma}
\newtheorem{prop}{Proposition}
\newtheorem{cor}{Corollary}
\newtheorem{defn}{Definition}

\newtheorem{rem}{Remark}

\begin{document}
		\title{A Bitopological Approach to Finite Reduction and Bounded
			Exact-Value Certificates for Fitting's Finite Heyting-valued Modal Logic}{}
		\author*[1,2]{\fnm{Litan Kumar} \sur{Das}}\email{ld06iitkgp@gmail.com}
		\author[2,3]{\fnm{Kumar Sankar} \sur{Ray}}\email{ksray@isical.ac.in}
		\author[1,2]{\fnm{Prakash Chandra} \sur{Mali}}\email{pcmali1959@gmail.com}
		\affil*[1]{\orgdiv{Department of Mathematics}, \orgname{Jadavpur University}, \orgaddress{\street{Jadavpur}, \city{Kolkata}, \postcode{700032}, \state{West Bengal}, \country{India}}}
		
		\affil[2]{\orgdiv{ECSU}, \orgname{Indian Statistical Institute, Kolkata}, \orgaddress{\street{} \city{Kolkata}, \postcode{700108}, \state{West Bengal}, \country{India}}}
		
		\affil[3]{\orgdiv{Department of Mathematics}, \orgname{Jadavpur University}, \orgaddress{\street{Jadavpur}, \city{Kolkata}, \postcode{700032}, \state{West Bengal}, \country{India}}}
		\abstract{Fitting's finite Heyting-valued modal logic interprets modal formulas
			over a finite Heyting algebra. We use a relational bitopological representation to obtain a finite-state reduction. For a finite model and a finite vocabulary, the modal subalgebra generated by the atomic valuations determines a state-evaluation map. We prove that the
			observational quotient is isomorphic to its finite image in the
			bitopological dual and that the quotient relation is the restriction
			of the canonical dual relation. Hence every formula over the
			vocabulary preserves its exact truth value, and the quotient is
			minimal among surjective reductions through which all generated
			observations factor. In addition, for any formula and state, we
			construct a finite tree-like exact-value certificate whose depth is
			bounded by modal depth and whose branching depends only on the height
			of the truth-value algebra and the number of boxed subformulas. Failed
			formulas therefore admit bounded reduced counterexamples preserving
			their precise failure values.}
\maketitle
\section{Introduction}
Many-valued modal logics extend ordinary modal logic by allowing
formulas to take values in an ordered algebra rather than only in the
two-element Boolean algebra. In Fitting's finite-valued modal logic,
truth values belong to a fixed finite Heyting algebra $\mathcal L$, and the
necessity operator $\Box$ is interpreted on a Kripke frame by taking the meet
of the truth values at all accessible states \cite{fitting1991many}. This
semantics preserves intermediate truth values and provides a natural
setting for modal information that may be graded, partially ordered,
incomplete, or uncertain. The algebraic and representation-theoretic foundations of Fitting's many-valued modal logic are already well developed. Maruyama studied the algebraic
semantics of lattice-valued modal logic and obtained topological and
natural-duality representations for the corresponding modal algebras
\cite{maruyama2009algebraic,maruyama2011dualities,maruyama2012natural}. These results extend the
classical representation theory of modal algebras
\cite{blackburn2001modal,hansoul1983duality,jonsson1951boolean}. A bitopological duality for the
non-modal algebras of Fitting's finite-valued logic was subsequently
developed in \cite{das2021bitopological}. Thus, the purpose of the present paper
is not to propose another general duality theory. We use a concise
relational bitopological representation to address a specific finite
reduction problem for Fitting models. Finite-state reduction is important in verification and model
checking because a system may contain different states that cannot be
distinguished by the observations relevant to a chosen specification
\cite{paige1987three,grumberg1999model}. More recent generic
partition-refinement methods provide uniform procedures for reducing
relational, weighted, probabilistic, and automata-based systems
\cite{wissmann2020efficient}.
General connections between logical observations, algebraic substructures, and minimal
state-space quotients have also been studied using duality
\cite{bezhanishvili2020minimisation}. The preservation of non-Boolean truth values is also important in
many-valued verification. Multi-valued symbolic model checking has
been developed for systems containing incomplete or inconsistent
information \cite{chechik2003multi}, and later work has considered
multi-valued and fuzzy temporal model checking over generalized
Kripke structures
\cite{li2019computation,pan2025incremental}. These works provide a
broader computational motivation for retaining intermediate truth
values, although they do not study the bitopological reduction
constructed here.\\
Tree-like model constructions are standard tools in modal logic
\cite{blackburn2001modal}. In model checking, tree-like
counterexamples have been studied for branching-time specifications
\cite{clarke2002tree}, and counterexample generation has also been
extended to multi-valued semantics
\cite{gurfinkel2003generating}. These results concern different
languages and verification problems. The certificate theorem proved
in Section \ref{5} instead exploits the finite height of $\mathcal L$ to bound the
number of successors needed to preserve the exact meet value of each
$\Box$-subformula.
In many-valued settings, quotient
constructions have been investigated through multi-valued
bisimulations and related reduction methods
\cite{hui2019multi,stankovic2023simulations}. Filtration and
finite-model constructions have likewise been considered for
many-valued modal systems \cite{lin2023many}. Recent work in a
different finite-valued modal framework has also established
completeness, the finite-model property, and decidability
\cite{karniel2024many}. Therefore, neither
finite-state reduction nor the general kernel-image method is new by
itself.
The problem considered here is more specific. Let
\[
M=(W,\mathcal R,V)
\]
be a finite $\mathcal L$-valued Kripke model, and let
$\Sigma\subseteq\mathsf{Prop}$ be a finite propositional vocabulary.
Inside the modal algebra $\mathcal L^W$, consider the $\mathcal L$-modal subalgebra \[
A_M^\Sigma
\]
generated by the atomic valuation functions $V(p)$, $p\in\Sigma$.
The elements of $A_M^\Sigma$ are exactly the semantic
$\mathcal L$-valued modal observations generated from the chosen vocabulary. Every state $x\in W$ determines an evaluation homomorphism
\[
\eta_M^\Sigma(x)\colon A_M^\Sigma\longrightarrow \mathcal L,
\qquad
\eta_M^\Sigma(x)(f)=f(x).
\]
Accordingly, two states $x,y\in W$ are observationally equivalent when
\[
\eta_M^\Sigma(x)=\eta_M^\Sigma(y),
\]
or, equivalently, when
\[
f(x)=f(y)
\qquad\text{for every }f\in A_M^\Sigma.
\]
Taking the kernel of an evaluation map and forming its quotient is a
standard construction. The new point of the present work is not this
set-theoretic quotient alone. The bitopological representation equips
\[
G(A_M^\Sigma)
=
\operatorname{Hom}_{\mathcal{VA}_{\mathcal L}}(A_M^\Sigma,\mathcal L)
\]
with two canonical topologies, a subalgebra-indexed structure map, and a canonical binary relation representing the modal operation. We prove that the observational quotient is realized exactly by the finite
image
\[
X_M^\Sigma=\eta_M^\Sigma[W]
\]
inside this dual space. More precisely, if $\mathcal R^{\mathrm{red}}$ is the
relation induced on $W/{\equiv_M^\Sigma}$ and
$\mathcal R_{A_M^\Sigma}$ is the canonical relation on the bitopological dual of $A_M^\Sigma$, then
\[
[x]\mathcal R^{\mathrm{red}}[y]
\quad\Longleftrightarrow\quad
\eta_M^\Sigma(x)
\mathcal R_{A_M^\Sigma}
\eta_M^\Sigma(y).
\]
Consequently, the natural bijection
\[
\overline{\eta}_M^\Sigma\colon
W/{\equiv_M^\Sigma}
\longrightarrow X_M^\Sigma,
\qquad
\overline{\eta}_M^\Sigma([x])
=
\eta_M^\Sigma(x),
\]
is an isomorphism of the corresponding finite relational
bitopological structures. This identification is the main
representation-specific contribution of the paper.\\
The purpose of the construction is to show that the bitopological
representation determines both parts of the reduction. The evaluation
map determines which states must be identified, while the canonical
dual relation determines the transition relation on the reduced
states. Thus the reduced model is not introduced as an arbitrary
quotient of the original Kripke model; it is recovered as the finite
relational image associated with the modal observation algebra
$A_M^\Sigma$.\\
Using this identification, we prove exact truth-value preservation:
for every formula
$\varphi\in\operatorname{Form}(\Sigma)$ and every $x\in W$,
\[
\llbracket\varphi\rrbracket_M(x)
=
\llbracket\varphi\rrbracket_{M^{\mathrm{red}}}([x]).
\]
The result preserves every element of $\mathcal L$, not only designated truth or truth at value $1$. Hence all intermediate truth values carried by the original model remain unchanged after reduction. We also establish a universal factorization property: among surjective reductions through which every member of $A_M^\Sigma$ factors, the constructed quotient has the least number of states. This minimality is relative to the chosen vocabulary and its generated modal observations.\\
A second result concerns formula-relative exact-value certificates.
Let $h(\mathcal L)$ denote the maximum number of elements in a chain
of $\mathcal L$, and let $m_{\Box}(\varphi)$ be the number of distinct
boxed subformulas of $\varphi$. For every $\mathcal L$-valued model
$M$, every state $x$, and every formula $\varphi$ of modal depth $d$,
we construct a finite rooted tree-like model
\[
\mathcal C(M,x,\varphi)
\]
such that
\[
\llbracket\varphi\rrbracket_{\mathcal C}(\varepsilon)
=
\llbracket\varphi\rrbracket_M(x).
\]
Every node of this certificate has at most
\[
b_{\mathcal L,\varphi}
=
\bigl(h(\mathcal L)-1\bigr)m_{\Box}(\varphi)
\]
successors, and consequently
\[
|W_{\mathcal C}|
\leq
1+b_{\mathcal L,\varphi}
+\cdots+
b_{\mathcal L,\varphi}^{\,d}.
\]
This bound is independent of the number of states of the original
model. If
$\llbracket\varphi\rrbracket_M(x)\neq1$, applying the observational
reduction to the certificate produces a finite reduced counterexample
preserving the same failure value.\\
The first main result identifies the finite
observational quotient of a Fitting model with the relational
bitopological image of its generated modal algebra. The second is a
height-sensitive exact-value certificate theorem whose size bound
depends only on $\mathcal L$ and the syntactic structure of the
formula, and not on the size of the original model.\\

The paper is organized as follows. Section \ref{2} recalls the
required algebraic, relational, and bitopological notions.
Section \ref{3} presents the relational bitopological representation
of $\mathcal L$-valued modal algebras. Section \ref{4} constructs the
finite observational quotient, identifies it with the finite dual
image, and proves relation compatibility, exact truth-value
preservation, and the associated minimality property. A finite
three-valued example illustrates the construction at the end of
Section \ref{4}. Section \ref{5} proves the meet-compression result,
constructs bounded exact-value certificates, and derives finite
reduced counterexamples. Section \ref{6} concludes the paper and
indicates directions for further work.

\section{Preliminaries}\label{2}
	We recall the algebraic, relational, and bitopological notions required
	in the sequel. Standard references for universal algebra, lattice
	theory, and category theory are
	\cite{burris1981course,davey2002introduction,adamek1990abstract}.\\
	Throughout the paper, \(\mathcal L\) denotes a fixed bounded finite
	distributive lattice with least element \(0\) and greatest element \(1\).
	Since \(\mathcal L\) is finite, for every \(a,b\in \mathcal L\) the
	element
	\[
	a\to b=\bigvee\{c\in \mathcal L: a\wedge c\leq b\}
	\]
	exists. Hence \(\mathcal L\) is a finite Heyting algebra. 
	 We write
	\[
	a\leftrightarrow b
	=
	(a\to b)\wedge(b\to a).
	\]
	\subsection{Fitting's finite-valued modal logic}
	We recall the basic syntax and algebraic semantics of Fitting's
	finite Heyting-valued modal logic \cite{fitting1991many} in the form used by Maruyama \cite{maruyama2009algebraic}. Let \(\mathsf{Prop}\) be a set of
	propositional variables. The language \(\mathcal L\)-\(\mathcal{VL}\)
	is generated by
	\[
	\varphi ::= p\mid 0\mid 1\mid
	\varphi\wedge\varphi\mid
	\varphi\vee\varphi\mid
	\varphi\to\varphi\mid T_\ell(\varphi),
	\]
	where \(p\in \mathsf{Prop}\) and \(\ell\in\mathcal L\). The connective
	\(T_\ell\) is an exact truth-value test. In the intended interpretation,
	\[
	T_\ell(a)=
	\begin{cases}
		1, & a=\ell,\\
		0, & a\neq \ell.
	\end{cases}
	\]
	Thus \(T_\ell(\varphi)\) expresses that the truth value of
	\(\varphi\) is exactly \(\ell\). For \(\ell\in\mathcal L\), we also use the derived operation
	\(U_\ell\), defined by
	\[
	U_\ell(a)=\bigvee\{T_{\ell'}(a):\ell\leq \ell',\ \ell'\in\mathcal L\}.
	\]
	Thus \(U_\ell(a)\) expresses that the truth value of \(a\) is at least
	\(\ell\).
\begin{defn}[\cite{maruyama2009algebraic}]
	\label{LVL}
An $\mathcal L\text{-}\mathcal{VL}$-algebra is an algebra
\[
\mathcal A
=
\bigl(A,\wedge,\vee,\to,0,1,(T_{\ell})_{\ell\in\mathcal L}\bigr)
\]
such that $(A,\wedge,\vee,\to,0,1)$ is a Heyting algebra and, for all
$a,b\in A$ and $\ell,\ell_1,\ell_2\in\mathcal L$, the following
conditions hold:
\begin{align*}
	&T_{\ell_1}(a)\wedge T_{\ell_2}(b)
	\leq
	T_{\ell_1\to\ell_2}(a\to b)
	\wedge T_{\ell_1\wedge\ell_2}(a\wedge b)
	\wedge T_{\ell_1\vee\ell_2}(a\vee b),\\
	&T_{\ell_2}(a)
	\leq
	T_{T_{\ell_1}(\ell_2)}
	\bigl(T_{\ell_1}(a)\bigr);
	\tag{L1}\\[1mm]
	&T_0(0)=1,\qquad
	T_{\ell}(0)=0\quad(\ell\neq0),\\
	&T_1(1)=1,\qquad
	T_{\ell}(1)=0\quad(\ell\neq1);
	\tag{L2}\\[1mm]
	&\bigvee_{\ell\in\mathcal L}T_{\ell}(a)=1,\\
	&T_{\ell}(a)\vee\bigl(T_{\ell}(a)\to0\bigr)=1,\\
	&T_{\ell_1}(a)\wedge T_{\ell_2}(a)=0
	\quad(\ell_1\neq\ell_2);
	\tag{L3}\\[1mm]
	&T_1\bigl(T_{\ell}(a)\bigr)=T_{\ell}(a),\\
	&T_0\bigl(T_{\ell}(a)\bigr)=T_{\ell}(a)\to0,\\
	&T_{\ell_2}\bigl(T_{\ell_1}(a)\bigr)=0
	\quad(\ell_2\notin\{0,1\});
	\tag{L4}\\[1mm]
	&T_1(a)\leq a,\qquad
	T_1(a\wedge b)=T_1(a)\wedge T_1(b);
	\tag{L5}\\[1mm]
	&\bigwedge_{\ell\in\mathcal L}
	\bigl(T_{\ell}(a)\leftrightarrow T_{\ell}(b)\bigr)
	\leq
	a\leftrightarrow b.
	\tag{L6}
\end{align*}
In \textup{(L1)}, the index $T_{\ell_1}(\ell_2)$ is computed in the
fixed algebra $\mathcal L$.
\end{defn}
A homomorphism of
\(\mathcal L\)-\(\mathcal{VL}\)-algebras is a map preserving
\(\wedge,\vee,\to,0,1\), and all operations \((T_\ell)_{\ell\in\mathcal L}\). We denote by $\mathcal{VA}_{\mathcal{L}}$ the category of $\mathcal{L}$-$\mathcal{VL}$-algebras and $\mathcal{L}$-$\mathcal{VL}$-algebras-homomorphisms.\\
The modal language \(\mathcal L\)-\(\mathcal{ML}\) is obtained from
\(\mathcal L\)-\(\mathcal{VL}\) by adding one unary modal operator
\(\Box\). No primitive possibility operator $\Diamond$ is assumed.\\
For a formula $\varphi$, let $\operatorname{Var}(\varphi)$ denote the
set of propositional variables occurring in $\varphi$, and let
$\operatorname{Sub}(\varphi)$ denote its set of subformulas. We write
\[
\operatorname{Sub}_{\Box}(\varphi)
=
\{\Box\psi\in\operatorname{Sub}(\varphi)\}
\]
and
\[
m_{\Box}(\varphi)
=
|\operatorname{Sub}_{\Box}(\varphi)|.
\]

The modal depth $\operatorname{md}(\varphi)$ is defined recursively by
\[
\begin{aligned}
	\operatorname{md}(p)
	&=
	\operatorname{md}(0)
	=
	\operatorname{md}(1)
	=
	0,\\
	\operatorname{md}(\varphi\circ\psi)
	&=
	\max\{\operatorname{md}(\varphi),
	\operatorname{md}(\psi)\},
	\qquad
	\circ\in\{\wedge,\vee,\to\},\\
	\operatorname{md}(T_{\ell}(\varphi))
	&=
	\operatorname{md}(\varphi),\\
	\operatorname{md}(\Box\varphi)
	&=
	\operatorname{md}(\varphi)+1.
\end{aligned}
\]
\begin{defn}[\cite{maruyama2009algebraic}]
	\label{LML}
	A $\mathcal{L}$-$\mathcal{ML}$-algebra is a $\mathcal{L}$-$\mathcal{VL}$-algebra $\mathcal A$ equipped with a unary connective $\Box: \mathcal{A}\to\mathcal{A}$ such that, for all \(a,b\in A\) and \(\ell\in\mathcal L\),
	\begin{enumerate}[(i)]
		\item $\Box(a\wedge b)=\Box a\wedge\Box b$;
		\item $\Box\bigl( U_{\ell}(a)\bigr)=U_{\ell}(\Box a)$, $\forall \ell\in\mathcal{L}$.
	\end{enumerate}
\end{defn}
A homomorphism of $\mathcal{L}$-$\mathcal{ML}$-algebras is a mapping that preserves all the operations of $\mathcal{L}$-$\mathcal{VL}$-algebras and the modal operation $\Box$. Let $\mathcal{MA}_{\mathcal{L}}$ denote the
category of $\mathcal{L}$-$\mathcal{ML}$-algebras and homomorphisms of $\mathcal{L}$-$\mathcal{ML}$-algebras.
\subsection{Finite \(\mathcal L\)-valued Kripke models}
Let $W$ be a non-empty set and let
\[
\mathcal R\subseteq W\times W
\]
be a binary relation. For $x\in W$, put
\[
\mathcal R[x]=\{y\in W:x\mathcal R y\}.
\]

An $\mathcal L$-valued Kripke model is a triple
\[
M=(W,\mathcal R,V),
\]
where
\[
V:\mathsf{Prop}\longrightarrow\mathcal L^W
\]
assigns an $\mathcal L$-valued function to every propositional
variable. The model is called finite when $W$ is finite.

The valuation extends recursively to a map
\[
\llbracket-\rrbracket_M:
\operatorname{Form}_{\mathcal L\text{-}\mathcal{ML}}
\longrightarrow
\mathcal L^W.
\]
The non-modal operations are interpreted pointwise. For the modal
operator,
\[
\llbracket\Box\varphi\rrbracket_M(x)
=
\bigwedge
\{\llbracket\varphi\rrbracket_M(y):
y\in\mathcal R[x]\}.
\]
Since $\mathcal L$ is finite, it is a complete lattice. Hence the meet
in the modal clause exists even when $\mathcal R[x]$ is infinite.
As usual,
\[
\bigwedge\varnothing=1.
\]
Equivalently, for every $f\in\mathcal L^W$, define
\[
(\Box_{\mathcal R}f)(x)
=
\bigwedge\{f(y):y\in\mathcal R[x]\}.
\]

\begin{lem}\label{lemma:function-algebra}
	With the pointwise $\mathcal L\text{-}\mathcal{VL}$-operations and
	$\Box_{\mathcal R}$, the algebra $\mathcal L^W$ is an
	$\mathcal L\text{-}\mathcal{ML}$-algebra.
\end{lem}
Thus the interpretation map
\[
\llbracket-\rrbracket_M
\]
is the unique $\mathcal L\text{-}\mathcal{ML}$-algebra homomorphism extending
the atomic valuation $V$.\\
For a state $x\in W$, we write
\[
M,x\models\varphi
\]
when
\[
\llbracket\varphi\rrbracket_M(x)=1.
\]
We write
\[
M\models\varphi
\]
when $M,x\models\varphi$ for every $x\in W$. If
\[
\llbracket\varphi\rrbracket_M(x)\neq1,
\]
then the pointed model $(M,x)$ is a counterexample to $\varphi$. Its
exact failure value is
\[
\llbracket\varphi\rrbracket_M(x).
\]
\subsection{Bitopological background and the non-modal duality}
A bitopological space is a triple
\[
(X,\tau_1,\tau_2),
\]
where \(\tau_1\) and \(\tau_2\) are topologies on \(X\). Let
\(\delta_i\) be the collection of \(\tau_i\)-closed sets, and put
\[
\beta_1=\tau_1\cap\delta_2,
\qquad
\beta_2=\tau_2\cap\delta_1.
\]
\begin{defn}[\cite{salbany1974bitopological}]
	\label{BT}
	\begin{enumerate}[(i)]
		\item A bitopological space $(X,\tau_1,\tau_2)$ is said to be pairwise Hausdorff space if for every pair $(x,y)$ of distinct points $x,y\in X$ there exist disjoint open sets $U_x\in\tau_1$ and $U_y\in\tau_2$ containing $x$ and $y$, respectively.
		\item A bitopological space $(X,\tau_1,\tau_2)$ is said to be pairwise zero-dimensional if $\beta_1$ is a basis for $\tau_1$ and $\beta_2$ is a basis for $\tau_2$.
		\item A bitopological space $(X,\tau_1,\tau_2)$ is said to be pairwise compact if the topological space $(X,\tau)$, where $\tau=\tau_1\vee\tau_2$, is compact.
	\end{enumerate}
\end{defn}
\begin{lem}\label{lem:finite-pairwise-discrete}
	Let $(X,\tau_1,\tau_2)$ be a finite pairwise Hausdorff
	bitopological space. Then both $\tau_1$ and $\tau_2$ are discrete.
\end{lem}
A \emph{pairwise Boolean space} is a pairwise Hausdorff, pairwise
zero-dimensional, and pairwise compact bitopological space. A map
\[
f:(X,\tau_1,\tau_2)\longrightarrow(Y,\sigma_1,\sigma_2)
\]
is pairwise continuous if it is continuous from $(X,\tau_i)$ to
$(Y,\sigma_i)$ for $i=1,2$. Pairwise Boolean spaces and pairwise
continuous maps form a category denoted by $\mathrm{PBS}$.\\
We now recall the category $\mathrm{PBS}_{\mathcal L}$ introduced in
\cite{das2021bitopological}. Let $\mathfrak S_{\mathcal L}$ be the set
of subalgebras of $\mathcal L$, and let $\Lambda_B$ denote the set of
subspaces of a pairwise Boolean space $B$ that are closed in both
topologies.

An object in $\mathrm{PBS}_{\mathcal L}$ is a pair
\[
(B,\alpha_B),
\]
where $B$ is a pairwise Boolean space and
\[
\alpha_B:\mathfrak S_{\mathcal L}\longrightarrow\Lambda_B
\]
satisfies
\[
\alpha_B(\mathcal L)=B
\]
and
\[
\alpha_B(\mathcal L_1\cap\mathcal L_2)
=
\alpha_B(\mathcal L_1)\cap\alpha_B(\mathcal L_2)
\]
for all $\mathcal L_1,\mathcal L_2\in\mathfrak S_{\mathcal L}$.

A morphism
\[
f:(B,\alpha_B)\longrightarrow(C,\alpha_C)
\]
is a pairwise continuous map such that
\[
f[\alpha_B(\mathcal L')]
\subseteq
\alpha_C(\mathcal L')
\]
for every $\mathcal L'\in\mathfrak S_{\mathcal L}$.

Equip $\mathcal L$ with the discrete topology in both coordinates and
define
\[
\alpha_{\mathcal L}(\mathcal L')=\mathcal L'
\qquad
(\mathcal L'\in\mathfrak S_{\mathcal L}).
\]
Then $(\mathcal L,\alpha_{\mathcal L})$ is an object in
$\mathrm{PBS}_{\mathcal L}$.

For an $\mathcal L\text{-}\mathcal{VL}$-algebra $\mathcal A$, put
\[
X_{\mathcal A}
=
\operatorname{Hom}_{\mathcal{VA}_{\mathcal L}}
(\mathcal A,\mathcal L).
\]
For $a\in A$, define
\[
\langle a\rangle
=
\{h\in X_{\mathcal A}:h(a)=1\}.
\]
Let $\tau_1^{\mathcal A}$ be the topology generated by the sets
$\langle a\rangle$, and let $\tau_2^{\mathcal A}$ be the topology
generated by their complements. Define
\[
\alpha_{\mathcal A}(\mathcal L')
=
\operatorname{Hom}_{\mathcal{VA}_{\mathcal L}}
(\mathcal A,\mathcal L')
\qquad
(\mathcal L'\in\mathfrak S_{\mathcal L}).
\]

The assignments
\[
F(B,\alpha_B)
=
\operatorname{Hom}_{\mathrm{PBS}_{\mathcal L}}
\bigl((B,\alpha_B),(\mathcal L,\alpha_{\mathcal L})\bigr)
\]
and
\[
G(\mathcal A)
=
\bigl(
X_{\mathcal A},
\tau_1^{\mathcal A},
\tau_2^{\mathcal A},
\alpha_{\mathcal A}
\bigr)
\]
extend contravariantly to morphisms by precomposition. They yield the
known dual equivalence
\[
\mathcal{VA}_{\mathcal L}
\simeq
\mathrm{PBS}_{\mathcal L}^{\mathrm{op}}.
\]

The corresponding evaluation maps are
\[
\gamma_{\mathcal A}:
\mathcal A\longrightarrow F(G(\mathcal A)),
\qquad
\gamma_{\mathcal A}(a)(h)=h(a),
\]
and
\[
\zeta_{(B,\alpha_B)}:
B\longrightarrow G(F(B,\alpha_B)),
\qquad
\zeta_{(B,\alpha_B)}(b)(f)=f(b).
\]
These maps are the components of the natural isomorphisms establishing
the dual equivalence. 
\section{Relational Bitopological Representation}\label{3}
This section develops the relational form of the bitopological
representation used later. We work within the natural duality
framework of Clark and Davey~\cite{clark1998natural}, together with
the modal adaptations for $\mathcal L$-valued modal algebras
developed in~\cite{maruyama2011dualities,maruyama2012natural}. The underlying
non-modal bitopological duality is taken from~\cite{das2021bitopological}.
The passage to the modal setting is obtained by equipping the dual
space with a canonical binary relation representing the operation
$\Box$.\\
Accordingly, the representation theorem established in this section is not
presented as a separate general duality contribution. Its role is
to provide, in a self-contained form, the canonical relation,
evaluation maps, and modal compatibility properties required for
the finite-reduction construction of Section~\ref{4}. The main results
of the paper are developed in Sections~\ref{4} and~\ref{5}.


Let $\mathcal R$ be a binary relation on a set $X$. For
$C\subseteq X$, put
\[
[\mathcal R]C
=
\{x\in X:\mathcal R[x]\subseteq C\}
\]
and
\[
\langle\mathcal R\rangle C
=
\{x\in X:\mathcal R[x]\cap C\neq\varnothing\},
\]
where
\[
\mathcal R[x]=\{y\in X:x\mathcal R y\}.
\]

\subsection{Relational $\mathcal L$-valued pairwise Boolean spaces}
Let $(X,\tau_1,\tau_2,\alpha_X)$ be an object of
$\mathrm{PBS}_{\mathcal L}$, and write
\[
\beta_1^X=\tau_1\cap\delta_2,
\qquad
\beta_2^X=\tau_2\cap\delta_1.
\]
\begin{defn}\label{def:prbsl}
	The category $\mathrm{PRBS}_{\mathcal L}$ is defined as follows.
	
	An object of $\mathrm{PRBS}_{\mathcal L}$ is a structure
	\[
	\mathfrak X
	=
	(X,\tau_1,\tau_2,\alpha_X,\mathcal R_X)
	\]
	such that $(X,\tau_1,\tau_2,\alpha_X)$ is an object of
	$\mathrm{PBS}_{\mathcal L}$ and $\mathcal R_X$ is a binary relation
	on $X$ satisfying:
	
	\begin{enumerate}
		\item[\textup{(R1)}]
		For every $x\in X$, the subspace $\mathcal R_X[x]$ is pairwise compact.

\item[\textup{(R2)}]
For every $C\in\beta_1^X$,
\[
[\mathcal R_X]C\in\beta_1^X
\qquad\text{and}\qquad
\langle\mathcal R_X\rangle C\in\beta_1^X.
\]

\item[\textup{(R3)}]
For every $\mathcal L'\in\mathfrak S_{\mathcal L}$ and every
$x\in\alpha_X(\mathcal L')$,
\[
\mathcal R_X[x]\subseteq\alpha_X(\mathcal L').
\]
\end{enumerate}
A morphism
\[
f:\mathfrak X\longrightarrow\mathfrak Y=(Y,\sigma_1,\sigma_2,\alpha_Y,\mathcal R_Y)
\]
in $\mathrm{PRBS}_{\mathcal L}$ is a
$\mathrm{PBS}_{\mathcal L}$-morphism between the underlying spaces
such that:
\begin{enumerate}
	\item[\textup{(Forth)}]
	if $x\mathcal R_Xx'$, then
	\[
	f(x)\mathcal R_Yf(x');
	\]
	
	\item[\textup{(Back)}]
	if $f(x)\mathcal R_Yy$, then there exists $x'\in X$ such that
\[
x\mathcal R_Xx'
\qquad\text{and}\qquad
f(x')=y.
\]
\end{enumerate}
\end{defn}
Since
\[
[\mathcal R](X\setminus C)
=
X\setminus\langle\mathcal R\rangle C
\]
and
\[
\langle\mathcal R\rangle(X\setminus C)
=
X\setminus[\mathcal R]C,
\]
condition \textup{(R2)} also gives
\[
[\mathcal R]D,\,
\langle\mathcal R\rangle D
\in\beta_2^X
\qquad
(D\in\beta_2^X).
\]
\subsection{The canonical relational space}
Let
\[
\mathcal A=(A,\Box)
\]
be an $\mathcal L$-$\mathcal{ML}$-algebra. Its underlying
bitopological dual is
\[
X_{\mathcal A}
=
\operatorname{Hom}_{\mathcal{VA}_{\mathcal L}}
(\mathcal A,\mathcal L),
\]
equipped with the topologies
$\tau_1^{\mathcal A},\tau_2^{\mathcal A}$ and the structure map
$\alpha_{\mathcal A}$ defined in Section \ref{2}.\\
Define a binary relation
$\mathcal R_{\mathcal A}$ on $X_{\mathcal A}$ by
\[
h\mathcal R_{\mathcal A}k
\]
if and only if, for every $a\in A$ and every
$\ell\in\mathcal L$,
\[
h(\Box a)\geq\ell
\quad\Longrightarrow\quad
k(a)\geq\ell.
\]
Equivalently,
\[
h\mathcal R_{\mathcal A}k
\quad\Longleftrightarrow\quad
h(\Box a)\leq k(a)
\quad\text{for every }a\in A.
\]
The following canonical relation identity is standard for
$\mathcal L$-$\mathcal{ML}$-algebras
\cite{maruyama2011dualities}.
\begin{lem}\label{lem:canonical-relation}
	For every $h\in X_{\mathcal A}$ and every $a\in A$,
	\[
	h(\Box a)
	=
	\bigwedge
	\{k(a):h\mathcal R_{\mathcal A}k\}.
	\]
	Here, as usual,
	\[
	\bigwedge\varnothing=1.
	\]
\end{lem}
Define
\[
G(\mathcal A)
=
\bigl(
X_{\mathcal A},
\tau_1^{\mathcal A},
\tau_2^{\mathcal A},
\alpha_{\mathcal A},
\mathcal R_{\mathcal A}
\bigr).
\]
For an $\mathcal L$-$\mathcal{ML}$-algebra homomorphism
\[
\psi:\mathcal A\longrightarrow\mathcal B,
\]
define
\[
G(\psi):G(\mathcal B)\longrightarrow G(\mathcal A)
\]
by
\[
G(\psi)(h)=h\circ\psi.
\]
\begin{prop}\label{prop:G-well-defined}
	The above assignments define a contravariant functor
	\[
	G:\mathcal{MA}_{\mathcal L}
	\longrightarrow
	\mathrm{PRBS}_{\mathcal L}.
	\]
	\begin{proof}
		By the non-modal duality recalled in Section~\ref{2},
		\[
		(X_{\mathcal A},
		\tau_1^{\mathcal A},
		\tau_2^{\mathcal A},
		\alpha_{\mathcal A},
		\mathcal R_{\mathcal A})
		\]
		is an object of $\mathrm{PBS}_{\mathcal L}$. It remains to verify
		conditions {\rm (R1)--(R3)} for the canonical relation
		$\mathcal R_{\mathcal A}$.
		
		For {\rm (R1)}, fix $h\in X_{\mathcal A}$. If
		$k\notin\mathcal R_{\mathcal A}[h]$, then for some $a\in A$,
		\[
		h(\Box a)\nleq k(a).
		\]
		Putting $\ell=h(\Box a)$, the set
		\[
		X_A\setminus\langle U_\ell(a)\rangle
		\]
		is a $\tau_2^{\mathcal A}$-open neighbourhood of $k$ disjoint from
		$\mathcal R_{\mathcal A}[h]$. Hence $\mathcal R_{\mathcal A}[h]$ is
		$\tau_2^{\mathcal A}$-closed and therefore pairwise compact.\\
		For {\rm (R2)}, every member of $\beta_1^{\mathcal A}$ is of the form
		$\langle a\rangle$ for some $a\in A$. By Lemma~\ref{lem:canonical-relation},
		\[
		[\mathcal R_{\mathcal A}]\langle a\rangle
		=\langle\Box a\rangle.
		\]
		Moreover, if $b=T_1(a)\to 0$, then
		\[
		\langle\mathcal R_{\mathcal A}\rangle\langle a\rangle
		=\left\langle T_1(\Box b)\to 0\right\rangle.
		\]
		Thus both sets belong to $\beta_1^{\mathcal A}$.
		
		For {\rm (R3)}, let
		$h\in\alpha_{\mathcal A}(\mathcal L')$ and $h\mathcal R_{\mathcal A}k$.
		If $k(a)=\ell\notin\mathcal L'$ for some $a\in A$, put
		$c=T_\ell(a)\to a$. Lemma~\ref{lem:canonical-relation} gives
		\[
		h(\Box c)=\ell,
		\]
		contradicting $h\in\alpha_A(\mathcal L')$. Hence
		\[
		\mathcal R_{\mathcal A}[h]\subseteq\alpha_{\mathcal A}(\mathcal L').
		\]
		Therefore $G(\mathcal A)$ is an object in $\mathrm{PRBS}_{\mathcal L}$.\\
		Now let $\psi:\mathcal A\to \mathcal B$ be an $\mathcal L$-$\mathcal{ML}$-algebra homomorphism.
		The non-modal duality shows that
		\[
		G(\psi):X_{\mathcal B}\to X_{\mathcal A},\qquad G(\psi)(h)=h\circ\psi,
		\]
		is a $\mathrm{PBS}_{\mathcal L}$-morphism. If $h\mathcal R_{\mathcal B}k$, then for every $a\in A$,
		\[
		G(\psi)(h)(\Box a)
		=h(\Box\psi(a))
		\leq k(\psi(a))
		=G(\psi)(k)(a),
		\]
		so the forth condition holds. For the back condition, suppose
		\[
		G(\psi)(h)\mathcal R_{\mathcal A}w.
		\]
		The set
		\[
		K=G(\psi)[\mathcal R_{\mathcal B}[h]]
		\]
		is pairwise compact. If $w\notin K$, pairwise Hausdorffness and
		compactness yield an element $c\in A$ such that
		\[
		z(c)=1\quad(z\in K),
		\qquad
		w(c)=0.
		\]
		Hence
		\[
		h(\Box\psi(c))=1
		\]
		by Lemma~\ref{lem:canonical-relation}, and therefore
		\[
		G(\psi)(h)(\Box c)=1.
		\]
		Since $G(\psi)(h)\mathcal R_{\mathcal A}w$, this implies $w(c)=1$, a
		contradiction. Thus $w\in K$, so there exists
		$k\in\mathcal R_{\mathcal B}[h]$ with $G(\psi)(k)=w$. Hence $G(\psi)$ is a morphism in
		$\mathrm{PRBS}_{\mathcal L}$. Identity and composition are
		preserved by precomposition, so $G$ is a contravariant functor.

	\end{proof}
\end{prop}
\subsection{The complex modal algebra}
Let
\[
\mathfrak X
=
(X,\tau_1,\tau_2,\alpha_X,\mathcal R_X)
\]
be an object in $\mathrm{PRBS}_{\mathcal L}$. Put
\[
F(\mathfrak X)
=
\operatorname{Hom}_{\mathrm{PBS}_{\mathcal L}}
\bigl(
(X,\tau_1,\tau_2,\alpha_X),
(\mathcal L,\alpha_{\mathcal L})
\bigr).
\]
The operations of the
$\mathcal L$-$\mathcal{VL}$-reduct are defined pointwise. Define the
modal operation by
\[
(\Box_{\mathcal R_X}f)(x)
=
\bigwedge
\{f(y):x\mathcal R_Xy\}.
\]
\begin{prop}\label{prop:F-well-defined}
	The algebra
	\[
	\bigl(F(\mathfrak X),\Box_{\mathcal R_X}\bigr)
	\]
	is an $\mathcal L$-$\mathcal{ML}$-algebra.
	
	Moreover, if
	\[
	g:\mathfrak X\longrightarrow\mathfrak Y=(Y,\sigma_1,\sigma_2,\alpha_Y,\mathcal R_Y)
	\]
	is a morphism in $\mathrm{PRBS}_{\mathcal L}$, then
\[
F(g):F(\mathfrak Y)\longrightarrow F(\mathfrak X),
\qquad
F(g)(f)=f\circ g,
\]
is an $\mathcal L$-$\mathcal{ML}$-algebra homomorphism.
Consequently,
\[
F:\mathrm{PRBS}_{\mathcal L}
\longrightarrow
\mathcal{MA}_{\mathcal L}
\]
is a contravariant functor.
\begin{proof}
	Let $f\in F(\mathfrak X)$ and, for $\ell\in\mathcal L$, put
	\[
	C_{\ell}(f)
	=
	\{x\in X:\ell\leq f(x)\}.
	\]
	Since
	\[
	C_{\ell}(f)
	=
	\bigl(U_{\ell}\circ f\bigr)^{-1}(\{1\}),
	\]
	the set $C_{\ell}(f)$ belongs to
	$\beta_1^X\cap\beta_2^X$. Furthermore,
	\[
	C_{\ell}(\Box_{\mathcal R_X}f)
	=
	[\mathcal R_X]C_{\ell}(f).
	\]
	Condition \textup{(R2)} therefore shows that the latter set also
	belongs to $\beta_1^X\cap\beta_2^X$. Since $\mathcal L$ is finite,
	each fibre of $\Box_{\mathcal R_X}f$ is a finite Boolean combination
	of these threshold sets. Thus $\Box_{\mathcal R_X}f$ is pairwise
	continuous. If $x\in\alpha_X(\mathcal L')$, then by \textup{(R3)},
	\[
	\mathcal R_X[x]\subseteq\alpha_X(\mathcal L').
	\]
	Since $f[\alpha_X(\mathcal L')]\subseteq\mathcal L'$ and
	$\mathcal L'$ is a finite subalgebra of $\mathcal L$, we obtain
	\[
	(\Box_{\mathcal R_X}f)(x)\in\mathcal L'.
	\]
	Hence $\Box_{\mathcal R_X}f\in F(\mathfrak X)$. The identities
	\[
	\Box_{\mathcal R_X}(f\wedge g)
	=
	\Box_{\mathcal R_X}f
	\wedge
	\Box_{\mathcal R_X}g
	\]
	and
	\[
	\Box_{\mathcal R_X}U_{\ell}(f)
	=
	U_{\ell}(\Box_{\mathcal R_X}f)
	\]
	follow directly from the definition. Finally, the forth and back conditions imply
	\[
	g[\mathcal R_X[x]]
	=
	\mathcal R_Y[g(x)].
	\]
	Therefore,
	\[
	\Box_{\mathcal R_X}(f\circ g)
	=
	(\Box_{\mathcal R_Y}f)\circ g,
	\]
	so precomposition preserves the modal operation.
\end{proof}
\end{prop}
\subsection{Representation theorem}
For an $\mathcal L$-$\mathcal{ML}$-algebra $\mathcal A$, define
\[
\gamma_{\mathcal A}:
\mathcal A\longrightarrow F(G(\mathcal A))
\]
by
\[
\gamma_{\mathcal A}(a)(h)=h(a).
\]

For
\[
\mathfrak X
=
(X,\tau_1,\tau_2,\alpha_X,\mathcal R_X),
\]
define
\[
\zeta_{\mathfrak X}:
\mathfrak X\longrightarrow G(F(\mathfrak X))
\]
by
\[
\zeta_{\mathfrak X}(x)(f)=f(x).
\]
\begin{thm}\label{thm:relational-representation}
	For every $\mathcal L$-$\mathcal{ML}$-algebra $\mathcal A$,
	$\gamma_{\mathcal A}$ is an isomorphism in
	$\mathcal{MA}_{\mathcal L}$. For every
	$\mathfrak X\in\mathrm{PRBS}_{\mathcal L}$,
	$\zeta_{\mathfrak X}$ is an isomorphism in
	$\mathrm{PRBS}_{\mathcal L}$. Consequently,
	\[
	\mathcal{MA}_{\mathcal L}
	\simeq
	\mathrm{PRBS}_{\mathcal L}^{\mathrm{op}}.
	\]
	\begin{proof}
		By the non-modal duality,
		$\gamma_{\mathcal A}$ is already an isomorphism of
		$\mathcal L$-$\mathcal{VL}$-algebras. For $a\in A$ and
		$h\in X_{\mathcal A}$, Lemma \ref{lem:canonical-relation} gives
		\[
		\begin{aligned}
			\bigl(
			\Box_{\mathcal R_{\mathcal A}}
			\gamma_{\mathcal A}(a)
			\bigr)(h)
			&=
			\bigwedge
			\{\gamma_{\mathcal A}(a)(k):
			h\mathcal R_{\mathcal A}k\}\\
			&=
			\bigwedge
			\{k(a):h\mathcal R_{\mathcal A}k\}\\
			&=
			h(\Box a)\\
			&=
			\gamma_{\mathcal A}(\Box a)(h).
		\end{aligned}
		\]
		Thus $\gamma_{\mathcal A}$ preserves $\Box$. Similarly, the non-modal duality shows that
		$\zeta_{\mathfrak X}$ is an isomorphism of the underlying
		$\mathcal L$-valued pairwise Boolean spaces. The definition of the
		canonical relation gives \[
		x\mathcal R_Xy
		\quad\Longrightarrow\quad
		\zeta_{\mathfrak X}(x)
		\mathcal R_{F(\mathfrak X)}
		\zeta_{\mathfrak X}(y).
		\]
		The converse follows from the separation property of the underlying
		bitopological duality and the pairwise compactness of
		$\mathcal R_X[x]$. Hence
		\[
		x\mathcal R_Xy
		\quad\Longleftrightarrow\quad
		\zeta_{\mathfrak X}(x)
		\mathcal R_{F(\mathfrak X)}
		\zeta_{\mathfrak X}(y).
		\]
		Therefore $\zeta_{\mathfrak X}$ and its inverse preserve the
		relations.The naturality of $\gamma$ and $\zeta$ follows immediately from their definitions and the action of the functors on morphisms by
		precomposition. Hence
		\[
		\mathcal{MA}_{\mathcal L}\simeq \mathrm{PRBS}_{\mathcal L}^{\mathrm{op}}.
		\]
	\end{proof}
\end{thm}
\section{Finite truth-value preserving reduction}\label{4}
In this section, let
\[
M=(W,\mathcal R,V)
\]
be a finite $\mathcal L$-valued Kripke model, and let
\[
\Sigma\subseteq\mathsf{Prop}
\]
be a finite propositional vocabulary. We construct a reduced model
that preserves the exact $\mathcal L$-truth value of every formula
over $\Sigma$. The construction is based on the modal algebra of
semantic observations generated by the variables in $\Sigma$. Unlike filtration, the present construction does not begin with a prescribed finite collection of formulas. The original model is
already finite, and the reduction preserves all formulas over the chosen vocabulary.
\subsection{The generated modal observation algebra}

Let
\[
\operatorname{Form}(\Sigma)
\]
denote the set of $\mathcal L$-$\mathcal{ML}$-formulas whose
propositional variables belong to $\Sigma$. The interpretation map
\[
\llbracket-\rrbracket_M:
\operatorname{Form}(\Sigma)\longrightarrow\mathcal L^W
\]
is defined as in Section \ref{2}.

\begin{defn}\label{def:generated-algebra}
	The modal observation algebra generated by $\Sigma$ in $M$ is the
	$\mathcal L$-$\mathcal{ML}$-subalgebra
	\[
	A_M^\Sigma
	=
	\operatorname{Sg}_{\mathcal{MA}_{\mathcal L}}
	\bigl(\{V(p):p\in\Sigma\}\bigr)
	\subseteq\mathcal L^W.
	\]
\end{defn}
Thus $A_M^\Sigma$ is the smallest subset of $\mathcal L^W$ containing
the atomic valuation functions $V(p)$, $p\in\Sigma$, and closed under
the pointwise $\mathcal L$-$\mathcal{VL}$-operations and the modal
operation $\Box_{\mathcal R}$.
\begin{lem}\label{lem:semantic-image}
The generated modal observation algebra is exactly the image of the
semantic interpretation map:
\[
A_M^\Sigma
=
\left\{
\llbracket\varphi\rrbracket_M:
\varphi\in\operatorname{Form}(\Sigma)
\right\}.
\]
	\begin{proof}
		Put
		\[
		B_M^\Sigma
		=
		\left\{
		\llbracket\varphi\rrbracket_M:
		\varphi\in\operatorname{Form}(\Sigma)
		\right\}.
		\]
		We prove that
		\[
		A_M^\Sigma=B_M^\Sigma.
		\]
		First, $B_M^\Sigma$ contains every generator $V(p)$, since
		\[
		V(p)=\llbracket p\rrbracket_M
		\qquad (p\in\Sigma).
		\]
		It also contains the constant functions $0$ and $1$. Moreover, if
		\[
		f=\llbracket\varphi\rrbracket_M
		\quad\text{and}\quad
		g=\llbracket\psi\rrbracket_M,
		\]
		then
		\[
		\begin{aligned}
			f\wedge g
			&=\llbracket\varphi\wedge\psi\rrbracket_M,\\
			f\vee g
			&=\llbracket\varphi\vee\psi\rrbracket_M,\\
			f\to g
			&=\llbracket\varphi\to\psi\rrbracket_M,
		\end{aligned}
		\]
		and, for every $\ell\in\mathcal L$,
		\[
		T_\ell(f)
		=
		\llbracket T_\ell(\varphi)\rrbracket_M.
		\]
		Similarly,
		\[
		\Box_{\mathcal R}f
		=
		\llbracket\Box\varphi\rrbracket_M.
		\]
		Hence $B_M^\Sigma$ is an
		$\mathcal L$-$\mathcal{ML}$-subalgebra of $\mathcal L^W$ containing
		all the generators $V(p)$, $p\in\Sigma$. By the minimality of the
		generated subalgebra,
		\[
		A_M^\Sigma\subseteq B_M^\Sigma.
		\]
		
		Conversely, we prove by structural induction on $\varphi$ that
		\[
		\llbracket\varphi\rrbracket_M\in A_M^\Sigma
		\]
		for every $\varphi\in\operatorname{Form}(\Sigma)$.
		
		If $\varphi=p\in\Sigma$, then
		\[
		\llbracket p\rrbracket_M=V(p)\in A_M^\Sigma.
		\]
		The constant functions $0$ and $1$ also belong to
		$A_M^\Sigma$. The inductive cases for
		$\wedge,\vee,\to,T_\ell$, and $\Box$ follow from the fact that
		$A_M^\Sigma$ is closed under all operations of an
		$\mathcal L$-$\mathcal{ML}$-algebra. Therefore,
		\[
		B_M^\Sigma\subseteq A_M^\Sigma.
		\]
		
		Consequently,
		\[
		A_M^\Sigma=B_M^\Sigma.
		\]
	\end{proof}
\end{lem}
Since both $W$ and $\mathcal L$ are finite,
\[
A_M^\Sigma\subseteq\mathcal L^W
\]
is a finite algebra.
\subsection{Observational equivalence and the evaluation image}
Every state $x\in W$ determines a map
\[
\eta_M^\Sigma(x):
A_M^\Sigma\longrightarrow\mathcal L
\]
defined by
\[
\eta_M^\Sigma(x)(f)=f(x).
\]
Since the operations of the
$\mathcal L$-$\mathcal{VL}$-reduct of $\mathcal L^W$ are defined
pointwise, $\eta_M^\Sigma(x)$ is an
$\mathcal L$-$\mathcal{VL}$-homomorphism. Thus
\[
\eta_M^\Sigma:
W\longrightarrow
X_{A_M^\Sigma},
\qquad
x\longmapsto\eta_M^\Sigma(x),
\]
where
\[
X_{A_M^\Sigma}
=
\operatorname{Hom}_{\mathcal{VA}_{\mathcal L}}
(A_M^\Sigma,\mathcal L)
\]
is the carrier of the bitopological dual
$G(A_M^\Sigma)$.
\begin{defn}\label{def:observational-equivalence}
For $x,y\in W$, define
\[
x\equiv_M^\Sigma y
\]
if and only if
\[
f(x)=f(y)
\qquad\text{for every }f\in A_M^\Sigma.
\]
We call $\equiv_M^\Sigma$ the observational equivalence induced by
$\Sigma$.
\end{defn}
It is immediate that $\equiv_M^\Sigma$ is an equivalence relation.
By Lemma~\ref{lem:semantic-image},
\[
\begin{aligned}
	x\equiv_M^\Sigma y
	&\quad\Longleftrightarrow\quad
	\llbracket\varphi\rrbracket_M(x)
	=
	\llbracket\varphi\rrbracket_M(y)
\quad	\text{for every }
	\varphi\in\operatorname{Form}(\Sigma).
\end{aligned}
\]
Thus two states are observationally equivalent precisely when every
formula over $\Sigma$ has the same exact $\mathcal L$-truth value at
those states. Moreover, by the definition of the evaluation map,
\[
\begin{aligned}
	x\equiv_M^\Sigma y
	&\quad\Longleftrightarrow\quad
	f(x)=f(y)
	\text{ for every }f\in A_M^\Sigma\\
	&\quad\Longleftrightarrow\quad
	\eta_M^\Sigma(x)=\eta_M^\Sigma(y).
\end{aligned}
\]

For a map $g:X\to Y$, its kernel equivalence is the relation
\[
\ker_{\mathrm{eq}}(g)
=
\{(x,x')\in X\times X:g(x)=g(x')\}.
\]
Consequently,
\[
\ker_{\mathrm{eq}}(\eta_M^\Sigma)
=
\{(x,y)\in W\times W:x\equiv_M^\Sigma y\},
\]
and hence
\[
\equiv_M^\Sigma
=
\ker_{\mathrm{eq}}(\eta_M^\Sigma).
\]
Thus $\equiv_M^\Sigma$ is the kernel equivalence of
$\eta_M^\Sigma$.
Put
\[
W_M^\Sigma=W/{\equiv_M^\Sigma}
\]
and let
\[
q_M^\Sigma:W\longrightarrow W_M^\Sigma,
\qquad
q_M^\Sigma(x)=[x],
\]
be the quotient map. Define the finite evaluation image by
\[
X_M^\Sigma
=
\eta_M^\Sigma[W]
=
\{\eta_M^\Sigma(x):x\in W\}
\subseteq X_{A_M^\Sigma}.
\]
Since $\eta_M^\Sigma$ has the same value on all members of an
$\equiv_M^\Sigma$-class, it induces a well-defined map
\[
\overline{\eta}_M^\Sigma:
W_M^\Sigma\longrightarrow X_M^\Sigma
\]
given by
\[
\overline{\eta}_M^\Sigma([x])
=
\eta_M^\Sigma(x).
\]
This map is bijective.\\

Equip $X_M^\Sigma$ with the subspace topologies
\[
\tau_i^{M,\Sigma}
=
\{U\cap X_M^\Sigma:
U\in\tau_i^{A_M^\Sigma}\},
\qquad i=1,2,
\]
and define
\[
\alpha_M^\Sigma(\mathcal L')
=
X_M^\Sigma
\cap
\alpha_{A_M^\Sigma}(\mathcal L')
\]
for every subalgebra $\mathcal L'\subseteq\mathcal L$. Finally, let
\[
\mathcal R_M^\Sigma
=
\mathcal R_{A_M^\Sigma}
\cap
(X_M^\Sigma\times X_M^\Sigma).
\]

\begin{prop}\label{prop:finite-dual-image}
	The structure
	\[
	\mathfrak X_M^\Sigma
	=
	\bigl(
	X_M^\Sigma,
	\tau_1^{M,\Sigma},
	\tau_2^{M,\Sigma},
	\alpha_M^\Sigma,
	\mathcal R_M^\Sigma
	\bigr)
	\]
	is an object in $\mathrm{PRBS}_{\mathcal L}$.
	\begin{proof}
	Since $W$ and $\mathcal L$ are finite, the algebra
	\[
	A_M^\Sigma\subseteq\mathcal L^W
	\]
	is finite. Therefore its dual carrier
	\[
	X_{A_M^\Sigma}
	=
	\operatorname{Hom}_{\mathcal{VA}_{\mathcal L}}
	(A_M^\Sigma,\mathcal L)
	\]
	is finite. The underlying bitopological dual is pairwise Hausdorff.
	Hence, by Lemma~\ref{lem:finite-pairwise-discrete}, both of its
	topologies are discrete. Their restrictions to $X_M^\Sigma$ are
	therefore also discrete. Consequently,
	\[
	\bigl(
	X_M^\Sigma,
	\tau_1^{M,\Sigma},
	\tau_2^{M,\Sigma}
	\bigr)
	\]
	is pairwise Hausdorff and pairwise zero-dimensional. It is also
	pairwise compact because $X_M^\Sigma$ is finite. Thus it is a
	pairwise Boolean space.	We next verify the conditions for the structure map. Since
	\[
	\alpha_{A_M^\Sigma}(\mathcal L)
	=
	X_{A_M^\Sigma},
	\]
	we obtain
	\[
	\alpha_M^\Sigma(\mathcal L)
	=
	X_M^\Sigma.
	\]
	Moreover, for subalgebras $\mathcal L_1,\mathcal L_2\subseteq\mathcal L$,
	\[
	\begin{aligned}
		\alpha_M^\Sigma(\mathcal L_1\cap\mathcal L_2)
		&=
		X_M^\Sigma\cap
		\alpha_{A_M^\Sigma}(\mathcal L_1\cap\mathcal L_2)\\
		&=
		X_M^\Sigma\cap
		\alpha_{A_M^\Sigma}(\mathcal L_1)\cap
		\alpha_{A_M^\Sigma}(\mathcal L_2)\\
		&=
		\alpha_M^\Sigma(\mathcal L_1)
		\cap
		\alpha_M^\Sigma(\mathcal L_2).
	\end{aligned}
	\]
	Thus the underlying structure is an object in
	$\mathcal{PBS}_{\mathcal L}$. It remains to verify the relational conditions. For every $h\in X_M^\Sigma$, the successor set
	\[
	R_M^\Sigma[h]
	\]
	is finite, and hence pairwise compact. Therefore condition
	\textup{(R1)} holds.
	
	Since both induced topologies are discrete, every subset of
	$X_M^\Sigma$ belongs to
	\[
	\beta_1^{M,\Sigma}
	=
	\tau_1^{M,\Sigma}\cap\delta_2^{M,\Sigma}.
	\]
	Hence, for every $C\in\beta_1^{M,\Sigma}$, \[
	[R_M^\Sigma]C
	\quad\text{and}\quad
	\langle R_M^\Sigma\rangle C
	\]
	also belong to $\beta_1^{M,\Sigma}$. Thus condition \textup{(R2)}
	holds.
	
	Finally, let $\mathcal L'\subseteq\mathcal L$, let
	\[
	h\in\alpha_M^\Sigma(\mathcal L'),
	\]
	and suppose that
	\[
	hR_M^\Sigma k.
	\]
	Then
	\[
	h\in\alpha_{A_M^\Sigma}(\mathcal L')
	\quad\text{and}\quad
	hR_{A_M^\Sigma}k.
	\]
	Since the canonical dual relation satisfies condition \textup{(R3)},
	we have
	\[
	k\in\alpha_{A_M^\Sigma}(\mathcal L').
	\]
	Also $k\in X_M^\Sigma$, by the definition of $R_M^\Sigma$.
	Therefore \[
	k\in
	X_M^\Sigma\cap
	\alpha_{A_M^\Sigma}(\mathcal L')
	=
	\alpha_M^\Sigma(\mathcal L').
	\]
	Thus condition \textup{(R3)} also holds. Hence $\mathfrak X_M^\Sigma$ is an object in
	$\mathcal{PRBS}_{\mathcal L}$.
	\end{proof}
\end{prop}
The map $\eta_M^\Sigma$ induces a bijection
\[
\overline{\eta}_M^\Sigma:
W_M^\Sigma\longrightarrow X_M^\Sigma
\]
defined by
\[
\overline{\eta}_M^\Sigma([x])
=
\eta_M^\Sigma(x).
\]
\subsection{The reduced relation and its dual characterization}
Define a binary relation
$\mathcal R^{\mathrm{red}}$ on $W_M^\Sigma$ by
\[
[x]\mathcal R^{\mathrm{red}}[y]
\]
if and only if there exist $x',y'\in W$ such that
\[
x'\equiv_M^\Sigma x,
\qquad
y'\equiv_M^\Sigma y,
\qquad
x'\mathcal R y'.
\]
For $f\in A_M^\Sigma$, define
\[
\bar f:W_M^\Sigma\longrightarrow\mathcal L
\]
by
\[
\bar f([x])=f(x).
\]
This is well defined by the definition of
$\equiv_M^\Sigma$.
\begin{lem}\label{lem:modal-quotient}
	For every $f\in A_M^\Sigma$,
	\[
	\overline{\Box_{\mathcal R}f}
	=
	\Box_{\mathcal R^{\mathrm{red}}}\bar f.
	\]
	\begin{proof}
		Let $[x]\in W_M^\Sigma$. Then
		\[
	\begin{aligned}
		\bigl(
		\Box_{\mathcal R^{\mathrm{red}}}\bar f
		\bigr)([x])
		&=
		\bigwedge
		\bigl\{
		\bar f([y]):
		[x]\mathcal R^{\mathrm{red}}[y]
		\bigr\}\\
		&=
		\bigwedge_{x'\in[x]}
		\bigwedge
		\bigl\{
		f(y):x'\mathcal R y
		\bigr\}\\
		&=
		\bigwedge_{x'\in[x]}
		(\Box_{\mathcal R}f)(x').
	\end{aligned}
	\]
	Since $A_M^\Sigma$ is closed under $\Box_{\mathcal R}$,
	\[
	\Box_{\mathcal R}f\in A_M^\Sigma.
	\]
	Therefore $\Box_{\mathcal R}f$ is constant on the class $[x]$, and
	hence
	\[
	\bigwedge_{x'\in[x]}
	(\Box_{\mathcal R}f)(x')
	=
	(\Box_{\mathcal R}f)(x).
	\]
	Thus
	\[
	\bigl(
	\Box_{\mathcal R^{\mathrm{red}}}\bar f
	\bigr)([x])
	=
	\overline{\Box_{\mathcal R}f}([x]).
	\]
	\end{proof}
\end{lem}
Lemma \ref{lem:modal-quotient} proves that the existentially induced relation is
compatible with every generated modal observation. The next result
identifies this relation with the canonical relation of the
bitopological dual.
\begin{prop}[Relation recovery]\label{prop:relation-recovery}
	For all $x,y\in W$,
	\[
	[x]\mathcal R^{\mathrm{red}}[y]
	\quad\Longleftrightarrow\quad
	\eta_M^\Sigma(x)
	\mathcal R_{A_M^\Sigma}
	\eta_M^\Sigma(y).
	\]
	\begin{proof}
		Suppose first that
		\[
		[x]\mathcal R^{\mathrm{red}}[y].
		\]
		Choose $x'\in[x]$ and $y'\in[y]$ such that
		\[
		x'\mathcal R y'.
		\]
		For every $f\in A_M^\Sigma$,
		\[
		(\Box_{\mathcal R}f)(x')
		=
		\bigwedge_{x'\mathcal R z}f(z)
		\leq f(y').
		\]
		Since $x'\in[x]$ and $y'\in[y]$, we have
		\[
		x'\equiv_M^\Sigma x
		\qquad\text{and}\qquad
		y'\equiv_M^\Sigma y.
		\]
		By definition of $\equiv_M^\Sigma$, every member of
		$A_M^\Sigma$ has the same value at observationally equivalent states.
		Since
		\[
		f,\Box_{\mathcal R}f\in A_M^\Sigma,
		\]
		it follows that
		\[
		(\Box_{\mathcal R}f)(x')
		=
		(\Box_{\mathcal R}f)(x)
		\qquad\text{and}\qquad
		f(y')=f(y).
		\] Hence
		\[
		(\Box_{\mathcal R}f)(x)
		\leq f(y).
		\]
		Equivalently,
		\[
		\eta_M^\Sigma(x)(\Box f)
		\leq
		\eta_M^\Sigma(y)(f)
		\]
		for every $f\in A_M^\Sigma$. By the definition of the canonical
		relation $\mathcal R_{A_M^\Sigma}$,
		\[
		\eta_M^\Sigma(x)
		\mathcal R_{A_M^\Sigma}
		\eta_M^\Sigma(y).
		\]
			Conversely, suppose that
		\[
		[x]\not\mathcal R^{\mathrm{red}}[y].
		\]
		Then no representative of $[x]$ has an
		$\mathcal R$-successor belonging to $[y]$. Put
		\[
		S
		=
		\bigcup_{x'\in[x]}\mathcal R[x'].
		\]
		Thus
		\[
		S\cap[y]=\varnothing.
		\]
		
		For every $z\in S$, we have
		\[
		z\not\equiv_M^\Sigma y.
		\]
		Hence there exists $f_z\in A_M^\Sigma$ such that
		\[
		f_z(z)\neq f_z(y).
		\]
		Define
		\[
		c_z
		=
		T_{f_z(z)}(f_z).
		\]
		Then
		\[
		c_z(z)=1
		\qquad\text{and}\qquad
		c_z(y)=0.
		\]
		Since $S$ is finite, put
		\[
		c
		=
		\bigvee_{z\in S}c_z.
		\]
		When $S=\varnothing$, the right-hand side is understood as the empty
		join $0$. In either case,
		\[
		c(y)=0
		\]
		and
		\[
		c(z)=1
		\qquad
		\text{for every }z\in S.
		\]
		For every $x'\in[x]$, all successors of $x'$ belong to $S$.
		Therefore
		\[
		(\Box_{\mathcal R}c)(x')=1.
		\]
		In particular,
		\[
		(\Box_{\mathcal R}c)(x)=1.
		\]
		Consequently,
		\[
		\eta_M^\Sigma(x)(\Box c)
		=
		1
		\not\leq
		0
		=
		\eta_M^\Sigma(y)(c).
		\]
		Thus
		\[
		\eta_M^\Sigma(x)
		\not\mathcal R_{A_M^\Sigma}
		\eta_M^\Sigma(y).
		\]
		This proves the converse implication.
	\end{proof}
\end{prop}
\begin{thm}[Dual-image reduction]\label{thm:dual-image-reduction}
The induced evaluation map
\[
\overline{\eta}_M^\Sigma:
W_M^\Sigma\longrightarrow X_M^\Sigma,
\qquad
\overline{\eta}_M^\Sigma([x])
=
\eta_M^\Sigma(x),
\]
is a bijection. Moreover, for all $x,y\in W$,
\[
[x]\mathcal R^{\mathrm{red}}[y]
\quad\Longleftrightarrow\quad
\overline{\eta}_M^\Sigma([x])\,
\mathcal R_M^\Sigma\,
\overline{\eta}_M^\Sigma([y]).
\]
Consequently, \[
\overline{\eta}_M^\Sigma:
\bigl(W_M^\Sigma,\mathcal R^{\mathrm{red}}\bigr)
\longrightarrow
\bigl(X_M^\Sigma,\mathcal R_M^\Sigma\bigr)
\]
is an isomorphism of relational structures.\\
Define, for $i=1,2$,
\[
\tau_i^{\mathrm{red}}
=
\left\{
(\overline{\eta}_M^\Sigma)^{-1}[U]:
U\in\tau_i^{M,\Sigma}
\right\},
\]
and, for every subalgebra $\mathcal L'\subseteq\mathcal L$, define
\[
\alpha_M^{\mathrm{red}}(\mathcal L')
=
(\overline{\eta}_M^\Sigma)^{-1}
\bigl[
\alpha_M^\Sigma(\mathcal L')
\bigr].
\]Then
\[
\overline{\eta}_M^\Sigma:
\bigl(
W_M^\Sigma,
\tau_1^{\mathrm{red}},
\tau_2^{\mathrm{red}},
\alpha_M^{\mathrm{red}},
\mathcal R^{\mathrm{red}}
\bigr)
\longrightarrow
\mathfrak X_M^\Sigma
\]
is an isomorphism in $\mathcal{PRBS}_{\mathcal L}$.
\begin{proof}
We first show that $\overline{\eta}_M^\Sigma$ is well defined.
Suppose that $[x]=[y]$. Then
\[
x\equiv_M^\Sigma y,
\]
and hence, by Definition~\ref{def:observational-equivalence},
\[
\eta_M^\Sigma(x)=\eta_M^\Sigma(y).
\]
Therefore,
\[
\overline{\eta}_M^\Sigma([x])
=
\overline{\eta}_M^\Sigma([y]).
\]
The map is surjective because
\[
X_M^\Sigma=\eta_M^\Sigma[W].
\]
To prove injectivity, suppose that
\[
\overline{\eta}_M^\Sigma([x])
=
\overline{\eta}_M^\Sigma([y]).
\]
Then
\[
\eta_M^\Sigma(x)=\eta_M^\Sigma(y).
\]
Since $\equiv_M^\Sigma$ is the kernel equivalence of
$\eta_M^\Sigma$, it follows that
\[
x\equiv_M^\Sigma y,
\]
and therefore $[x]=[y]$. Thus
$\overline{\eta}_M^\Sigma$ is a bijection.

By Proposition~\ref{prop:relation-recovery}, \[
[x]\mathcal R^{\mathrm{red}}[y]
\quad\Longleftrightarrow\quad
\eta_M^\Sigma(x)\,
\mathcal R_{A_M^\Sigma}\,
\eta_M^\Sigma(y).
\]
Since
\[
\overline{\eta}_M^\Sigma([x])
=
\eta_M^\Sigma(x),
\qquad
\overline{\eta}_M^\Sigma([y])
=
\eta_M^\Sigma(y),
\]
and $\mathcal R_M^\Sigma$ is the restriction of
$\mathcal R_{A_M^\Sigma}$ to $X_M^\Sigma$, we obtain \[
[x]\mathcal R^{\mathrm{red}}[y]
\quad\Longleftrightarrow\quad
\overline{\eta}_M^\Sigma([x])\,
\mathcal R_M^\Sigma\,
\overline{\eta}_M^\Sigma([y]).
\]
Hence $\overline{\eta}_M^\Sigma$ is an isomorphism of relational
structures.\\

By the definitions of $\tau_i^{\mathrm{red}}$, the map
$\overline{\eta}_M^\Sigma$ is a homeomorphism between
$(W_M^\Sigma,\tau_i^{\mathrm{red}})$ and
$(X_M^\Sigma,\tau_i^{M,\Sigma})$ for $i=1,2$. Moreover,
\[
\overline{\eta}_M^\Sigma
\bigl[
\alpha_M^{\mathrm{red}}(\mathcal L')
\bigr]
=
\alpha_M^\Sigma(\mathcal L')
\]
for every subalgebra $\mathcal L'\subseteq\mathcal L$. Together with the preservation and reflection of the relations, this
shows that $\overline{\eta}_M^\Sigma$ is an isomorphism in
$\mathcal{PRBS}_{\mathcal L}$.
\end{proof}
\end{thm}
Thus the reduced relational structure is not chosen independently of
the duality. Its states are the evaluation points determined by the
original model, and its relation is the restriction of the canonical
dual relation to those points.
\subsection{Exact truth-value preservation}

Define the reduced $\Sigma$-model
\[
M_\Sigma^{\mathrm{red}}
=
\bigl(
W_M^\Sigma,
\mathcal R^{\mathrm{red}},
V^{\mathrm{red}}
\bigr),
\]
where, for every $p\in\Sigma$,
\[
V^{\mathrm{red}}(p)([x])
=
V(p)(x).
\]
This definition is independent of the representative. Indeed, if
$[x]=[y]$, then $x\equiv_M^\Sigma y$. Since
\[
V(p)\in A_M^\Sigma,
\]
Definition~\ref{def:observational-equivalence} gives
\[
V(p)(x)=V(p)(y).
\]
Equivalently,
\[
V^{\mathrm{red}}(p)=\overline{V(p)}.
\]
\begin{thm}[Exact preservation]\label{thm:exact-preservation}
For every $\varphi\in\operatorname{Form}(\Sigma)$,
\[
\llbracket\varphi\rrbracket_{M_\Sigma^{\mathrm{red}}}
=
\overline{\llbracket\varphi\rrbracket_M}.
\]
Equivalently, for every $x\in W$,
\[
\llbracket\varphi\rrbracket_
{M_\Sigma^{\mathrm{red}}}([x])
=
\llbracket\varphi\rrbracket_M(x).
\]
In terms of the quotient map
\[
q_M^\Sigma:W\longrightarrow W_M^\Sigma,
\]
the result can also be written as \[
\llbracket\varphi\rrbracket_M
=
\llbracket\varphi\rrbracket_
{M_\Sigma^{\mathrm{red}}}
\circ q_M^\Sigma.
\]
\begin{proof}
We proceed by structural induction on $\varphi$.

If $\varphi=p\in\Sigma$, then
\[
\llbracket p\rrbracket_{M_\Sigma^{\mathrm{red}}}
=
V^{\mathrm{red}}(p)
=
\overline{V(p)}
=
\overline{\llbracket p\rrbracket_M}.
\]
The cases $\varphi=0$ and $\varphi=1$ are immediate. Suppose that the result holds for $\psi$ and $\chi$. Since the
non-modal operations are interpreted pointwise, we obtain, for
$\circ\in\{\wedge,\vee,\to\}$,
\[
\begin{aligned}
	\llbracket\psi\circ\chi\rrbracket_
	{M_\Sigma^{\mathrm{red}}}
	&=
	\llbracket\psi\rrbracket_
	{M_\Sigma^{\mathrm{red}}}
	\circ
	\llbracket\chi\rrbracket_
	{M_\Sigma^{\mathrm{red}}}\\
	&=
	\overline{\llbracket\psi\rrbracket_M}
	\circ
	\overline{\llbracket\chi\rrbracket_M}\\
	&=
	\overline{
		\llbracket\psi\rrbracket_M
		\circ
		\llbracket\chi\rrbracket_M
	}\\
	&=
	\overline{
		\llbracket\psi\circ\chi\rrbracket_M
	}.
\end{aligned}
\]
Similarly, for every $\ell\in\mathcal L$,
\[
\begin{aligned}
	\llbracket T_\ell(\psi)\rrbracket_
	{M_\Sigma^{\mathrm{red}}}
	&=
	T_\ell\bigl(
	\llbracket\psi\rrbracket_
	{M_\Sigma^{\mathrm{red}}}
	\bigr)\\
	&=
	T_\ell\bigl(
	\overline{\llbracket\psi\rrbracket_M}
	\bigr)\\
	&=
	\overline{
		T_\ell(\llbracket\psi\rrbracket_M)
	}\\
	&=
	\overline{
		\llbracket T_\ell(\psi)\rrbracket_M
	}.
\end{aligned}
\]
Finally, suppose that $\varphi=\Box\psi$. By
Lemma~\ref{lem:semantic-image},
\[
\llbracket\psi\rrbracket_M\in A_M^\Sigma.
\]
Hence Lemma~\ref{lem:modal-quotient} applies. Using the induction
hypothesis, we obtain
\[
\begin{aligned}
	\llbracket\Box\psi\rrbracket_
	{M_\Sigma^{\mathrm{red}}}
	&=
	\Box_{\mathcal R^{\mathrm{red}}}
	\left(
	\llbracket\psi\rrbracket_
	{M_\Sigma^{\mathrm{red}}}
	\right)\\
	&=
	\Box_{\mathcal R^{\mathrm{red}}}
	\left(
	\overline{\llbracket\psi\rrbracket_M}
	\right)\\
	&=
	\overline{
		\Box_{\mathcal R}
		\llbracket\psi\rrbracket_M
	}\\
	&=
	\overline{
		\llbracket\Box\psi\rrbracket_M
	}.
\end{aligned}
\]
This completes the induction.
\end{proof}
\end{thm}
Thus the reduction preserves the actual value of every formula in
$\mathcal L$, not merely whether its value is $1$. In particular, for
every $\ell\in\mathcal L$,
\[
\ell\leq\llbracket\varphi\rrbracket_M(x)
\quad\Longleftrightarrow\quad
\ell\leq
\llbracket\varphi\rrbracket_
{M_\Sigma^{\mathrm{red}}}([x]).
\]
\subsection{Verification on the reduced model}

Let $\Sigma$ contain all propositional variables occurring in
$\varphi$. Since the reduction preserves exact truth values, model
checking may be carried out on the reduced model without changing the
value of $\varphi$.

\begin{cor}[Verification consequence]
	\label{prop:reduced-verification}
	Let $M=(W,\mathcal R,V)$ be a finite $\mathcal L$-valued Kripke model.
	Then
	\[
	M\models\varphi
	\quad\Longleftrightarrow\quad
	M_\Sigma^{\mathrm{red}}\models\varphi,
	\]
	where
	\[
	M\models\varphi
	\]
	means that
	\[
	\llbracket\varphi\rrbracket_M(x)=1
	\qquad\text{for every }x\in W.
	\]
	More generally, for every $x\in W$ and every $\ell\in\mathcal L$,
	\[
	\ell\leq\llbracket\varphi\rrbracket_M(x)
	\quad\Longleftrightarrow\quad
	\ell\leq
	\llbracket\varphi\rrbracket_{M_\Sigma^{\mathrm{red}}}([x]).
	\]
	\begin{proof}
		By Theorem~\ref{thm:exact-preservation},
		\[
		\llbracket\varphi\rrbracket_M(x)
		=
		\llbracket\varphi\rrbracket_{M_\Sigma^{\mathrm{red}}}([x])
		\]
		for every $x\in W$. Hence the threshold equivalence follows
		immediately. For the first assertion, if $M\models\varphi$, then every quotient
		state $[x]\in W_M^\Sigma$ has value $1$, since the quotient map
		$q_M^\Sigma:W\twoheadrightarrow W_M^\Sigma$ is surjective. The
		converse follows from the same equality.
	\end{proof}
\end{cor}

\subsection{The universal property of the quotient}\label{subsec:universal-property}
The following minimality property is relative to the observation
algebra $A_M^\Sigma$. It concerns surjective maps of state sets through
which all generated observations factor; no relational or
bitopological compatibility of these maps is assumed.
\begin{defn}\label{def:observation-preserving-surjection}
Let $Y$ be a set and let
\[
h:W\twoheadrightarrow Y
\]
be a surjection. We say that $h$ preserves the generated
$\Sigma$-observations if, for every $f\in A_M^\Sigma$, there exists a
map
\[
f_h:Y\longrightarrow\mathcal L
\]
such that
\[
f=f_h\circ h.
\]
Since $h$ is surjective, the map $f_h$, when it exists, is unique.
\end{defn}
The factorization condition has the following equivalent
interpretation.
\begin{lem}\label{lem:observation-fibres}
	A surjection $h:W\twoheadrightarrow Y$ preserves the generated
	$\Sigma$-observations if and only if
	\[
	h(x)=h(y)
	\quad\Longrightarrow\quad
	x\equiv_M^\Sigma y
	\]
	for all $x,y\in W$. Equivalently,
	\[
	\ker_{\mathrm{eq}}(h)
	\subseteq
	\equiv_M^\Sigma.
	\]
	\begin{proof}
		Suppose first that $h$ preserves the generated observations and that
		$h(x)=h(y)$. For every $f\in A_M^\Sigma$, choose
		$f_h:Y\to\mathcal L$ such that $f=f_h\circ h$. Then
		\[
		f(x)
		=
		f_h(h(x))
		=
		f_h(h(y))
		=
		f(y).
		\]
		Hence $x\equiv_M^\Sigma y$. Conversely, suppose that
		\[
		h(x)=h(y)
		\quad\Longrightarrow\quad
		x\equiv_M^\Sigma y.
		\]
		For $f\in A_M^\Sigma$, define
		\[
		f_h(h(x))=f(x).
		\]
		This is well defined: if $h(x)=h(y)$, then
		$x\equiv_M^\Sigma y$, and therefore $f(x)=f(y)$. By construction, \[
		f=f_h\circ h.
		\]
		Thus $h$ preserves every generated observation.
		\end{proof}
		\end{lem}
		The quotient map
		\[
		q_M^\Sigma:W\twoheadrightarrow W_M^\Sigma
		\]
		itself preserves the generated observations. Indeed, for every
		$f\in A_M^\Sigma$, the map
		\[
		\bar f:W_M^\Sigma\longrightarrow\mathcal L,
		\qquad
		\bar f([x])=f(x),
		\]
		is well defined and satisfies
		\[
		f=\bar f\circ q_M^\Sigma.
		\]
\begin{thm}[Universal factorization]\label{thm:minimality}
Let
\[
h:W\twoheadrightarrow Y
\]
be a surjection preserving the generated $\Sigma$-observations. Then
there exists a unique surjective map
\[
k:Y\twoheadrightarrow W_M^\Sigma
\]
such that
\[
q_M^\Sigma=k\circ h.
\]
Consequently,
\[
|W_M^\Sigma|\leq |Y|.
\]
\begin{proof}
Define
\[
k(h(x))=[x].
\]
We first show that $k$ is well defined. Suppose that
\[
h(x)=h(y).
\]
Since $h$ preserves the generated observations,
Lemma~\ref{lem:observation-fibres} gives
\[
x\equiv_M^\Sigma y.
\]
Hence
\[
[x]=[y],
\]
so the definition of $k$ is independent of the chosen representative
of an element of $Y$. For every $x\in W$,
\[
(k\circ h)(x)
=
k(h(x))
=
[x]
=
q_M^\Sigma(x).
\]
Therefore,
\[
q_M^\Sigma=k\circ h.
\]
The map $k$ is surjective because, for every class
$[x]\in W_M^\Sigma$,
\[
[x]=k(h(x)).
\]

Finally, suppose that another map
\[
k':Y\longrightarrow W_M^\Sigma
\]
satisfies
\[
q_M^\Sigma=k'\circ h.
\]
Every element of $Y$ has the form $h(x)$ because $h$ is surjective. Therefore,
\[
k'(h(x))
=
q_M^\Sigma(x)
=
k(h(x))
\]
for every $x\in W$. Hence $k'=k$.\\
The inequality
\[
|W_M^\Sigma|\leq |Y|
\]
follows from the surjectivity of $k$.
\end{proof}
\end{thm}
\begin{cor}\label{cor:minimal-exact-reduction}
Let
\[
N=(Y,\mathcal S,U)
\]
be an $\mathcal L$-valued Kripke model over $\Sigma$, and let
\[
h:W\twoheadrightarrow Y
\]
be a surjection satisfying
\[
\llbracket\varphi\rrbracket_M
=
\llbracket\varphi\rrbracket_N\circ h
\]
for every $\varphi\in\operatorname{Form}(\Sigma)$. Then
\[
|W_M^\Sigma|\leq |Y|.
\]
	\begin{proof}
	Let $f\in A_M^\Sigma$. By
	Lemma~\ref{lem:semantic-image}, there exists
	$\varphi\in\operatorname{Form}(\Sigma)$ such that
	\[
	f=\llbracket\varphi\rrbracket_M.
	\]
	By assumption,
	\[
	f
	=
	\llbracket\varphi\rrbracket_N\circ h.
	\]
	Hence $h$ preserves every generated $\Sigma$-observation.
	The conclusion follows from
	Theorem  \ref{thm:minimality}.

	\end{proof}
\end{cor}
\begin{rem}\label{rem:minimality-scope}
The preceding result is a relative minimality statement. It does not
assert that $M_\Sigma^{\mathrm{red}}$ is a smallest
$\mathcal L$-valued model satisfying the same formulas as $M$, nor
that it is a smallest countermodel to any particular formula. It
states that, among surjective images of the state set $W$ preserving
the exact values of all formulas over $\Sigma$, the quotient
$W_M^\Sigma$ has the least number of states. Equivalently,
$\equiv_M^\Sigma$ is the largest equivalence relation on $W$ through
which all generated $\Sigma$-observations factor.
\end{rem}

\subsection{Finite construction}\label{subsec:finite-construction}
We describe a finite procedure for constructing the observation
algebra and the reduced model. Let $\mathbf 0,\mathbf 1\in\mathcal L^W$
denote the constant functions with values $0$ and $1$, respectively,
and put
\[
A_0
=
\{\mathbf 0,\mathbf 1\}
\cup
\{V(p):p\in\Sigma\}.
\]
Having defined $A_n$, put
\[
\begin{aligned}
	A_{n+1}
	={}&A_n\\
	&\cup
	\{f\wedge g,\;f\vee g,\;f\to g:
	f,g\in A_n\}\\
	&\cup
	\{T_\ell(f):
	f\in A_n,\ \ell\in\mathcal L\}\\
	&\cup
	\{\Box_{\mathcal R}f:
	f\in A_n\}.
\end{aligned}
\]
Thus
\[
A_0\subseteq A_1\subseteq A_2\subseteq\cdots
\subseteq\mathcal L^W.
\]
Since $W$ and $\mathcal L$ are finite,
\[
|\mathcal L^W|
=
|\mathcal L|^{|W|}
\]
is finite. Hence there exists $N\in\mathbb N$ such that
\[
A_N=A_{N+1}.
\]
We claim that
\[
A_N=A_M^\Sigma.
\]
Indeed, $A_0\subseteq A_M^\Sigma$, and an induction on $n$ gives
\[
A_n\subseteq A_M^\Sigma
\]
for every $n$. Hence
\[
A_N\subseteq A_M^\Sigma.
\]
Conversely, the equality $A_N=A_{N+1}$ shows that $A_N$ is closed
under all operations of an $\mathcal L$-$\mathcal{ML}$-algebra.
Since it contains every generator $V(p)$, $p\in\Sigma$, the
minimality of the generated subalgebra gives
\[
A_M^\Sigma\subseteq A_N.
\]
Therefore,
\[
A_N=A_M^\Sigma.
\]
Now enumerate
\[
A_M^\Sigma=\{f_1,\ldots,f_m\}.
\]
For each $x\in W$, define its evaluation vector by
\[
\operatorname{ev}_M^\Sigma(x)
=
\bigl(f_1(x),\ldots,f_m(x)\bigr)
\in\mathcal L^m.
\]
Then
\[
x\equiv_M^\Sigma y
\quad\Longleftrightarrow\quad
\operatorname{ev}_M^\Sigma(x)
=
\operatorname{ev}_M^\Sigma(y).
\]
Thus the state set $W_M^\Sigma$ is obtained by grouping together
states having identical evaluation vectors.

Finally, the relation $\mathcal R^{\mathrm{red}}$ is computed by
\[
[x]\mathcal R^{\mathrm{red}}[y]
\]
if and only if there exist $x'\in[x]$ and $y'\in[y]$ such that
\[
x'\mathcal R y',
\]
and the reduced valuation is given by
\[
V^{\mathrm{red}}(p)([x])=V(p)(x)
\qquad(p\in\Sigma).
\]
This yields a finite effective construction of
$M_\Sigma^{\mathrm{red}}$. No claim of an optimal implementation or
complexity bound is made.

\subsection{A finite three-valued example}
\label{subsec:finite-example}

We illustrate that observationally equivalent states need not have
identical successor sets. Let
\[
\mathcal L=\{0<a<1\}
\]
be the three-element Heyting chain, let
\[
W=\{x,y,u,v\},
\]
and take
\[
\Sigma=\{p\}.
\]
Define
\[
\mathcal R[x]=\{u\},
\qquad
\mathcal R[y]=\{v\},
\qquad
\mathcal R[u]=\mathcal R[v]=\varnothing,
\]
and let
\[
V(p)(x)=V(p)(y)=a,
\qquad
V(p)(u)=V(p)(v)=1.
\]

Although
\[
\mathcal R[x]\neq\mathcal R[y],
\]
the states $x$ and $y$ are observationally equivalent. To see this,
put
\[
B=
\left\{
f\in\mathcal L^W:
f(x)=f(y)
\text{ and }
f(u)=f(v)
\right\}.
\]
The set $B$ contains the constant functions and $V(p)$, and it is
closed under all pointwise operations of an
$\mathcal L$-$\mathcal{VL}$-algebra.

It is also closed under $\Box_{\mathcal R}$. Indeed, for every
$f\in B$,
\[
\begin{aligned}
	(\Box_{\mathcal R}f)(x)
	&=f(u),\\
	(\Box_{\mathcal R}f)(y)
	&=f(v),
\end{aligned}
\]
and these values are equal because $f(u)=f(v)$. Moreover,
\[
(\Box_{\mathcal R}f)(u)
=
(\Box_{\mathcal R}f)(v)
=
1,
\]
since $u$ and $v$ have no successors. Thus
\[
\Box_{\mathcal R}f\in B.
\]

Consequently, $B$ is an $\mathcal L$-$\mathcal{ML}$-subalgebra of
$\mathcal L^W$ containing the generator $V(p)$. By the minimality of
the generated observation algebra,
\[
A_M^\Sigma\subseteq B.
\]
Therefore every $f\in A_M^\Sigma$ satisfies
\[
f(x)=f(y)
\qquad\text{and}\qquad
f(u)=f(v).
\]
Hence
\[
x\equiv_M^\Sigma y
\qquad\text{and}\qquad
u\equiv_M^\Sigma v.
\]

These are the only identifications, since
\[
V(p)(x)=a\neq1=V(p)(u).
\]
Thus the observational equivalence classes are
\[
[x]=[y]=\{x,y\},
\qquad
[u]=[v]=\{u,v\}.
\]

The reduced state set is therefore
\[
W_M^\Sigma
=
\bigl\{\{x,y\},\{u,v\}\bigr\}.
\]
Its relation is
\[
\{x,y\}\,
\mathcal R^{\mathrm{red}}\,
\{u,v\},
\]
and $\{u,v\}$ has no successor. The reduced valuation is
\[
V^{\mathrm{red}}(p)(\{x,y\})=a,
\qquad
V^{\mathrm{red}}(p)(\{u,v\})=1.
\]

Furthermore,
\[
\eta_M^\Sigma(x)=\eta_M^\Sigma(y),
\qquad
\eta_M^\Sigma(u)=\eta_M^\Sigma(v).
\]
Hence the evaluation image $X_M^\Sigma$ also has two points.
Theorem~\ref{thm:dual-image-reduction} identifies its restricted
canonical relation with $\mathcal R^{\mathrm{red}}$, while
Theorem~\ref{thm:exact-preservation} ensures that every formula over
$p$ has the same exact $\mathcal L$-truth value in the original and
reduced models.
\section{Bounded exact-value certificates and counterexample extraction}
\label{5}
Tree-like counterexamples have previously been studied in classical
model checking, particularly for branching-time specifications
\cite{clarke2002tree}. Counterexample generation has also been
developed for multi-valued CTL over De Morgan algebras
\cite{gurfinkel2003generating} and for multi-valued bounded model
checking \cite{andrade2008direct}. The result below addresses a
different problem. For Fitting's finite Heyting-valued modal logic, it
uses the finite height of $\mathcal L$ to bound the successors needed
to preserve the meet defining each boxed subformula. The resulting
tree-like certificate preserves the exact $\mathcal L$-truth value and
has a size bound independent of the size of the original model.\\
The reduction developed in Section \ref{4} preserves every generated
modal observation of a given finite model. We now consider a
formula-dependent construction of a different kind. For a formula
$\varphi$ evaluated at a state $x$, we extract a finite tree-like
model that preserves the exact value of $\varphi$ at $x$. Its depth is
bounded by the modal depth of $\varphi$, while its branching is bounded
in terms of the height of the truth-value algebra and the number of
boxed subformulas of $\varphi$. The resulting certificate may then be
further compressed by the reduction of Section \ref{4}.

\subsection{Finite witnesses for lattice meets}

Let
\[
h(\mathcal L)
\]
denote the maximum number of elements in a chain of the finite Heyting
algebra $\mathcal L$.

\begin{lem}[Meet compression]\label{lem:meet-compression}
For every $S\subseteq\mathcal L$, there exists
$S_0\subseteq S$ such that
\[
|S_0|\leq h(\mathcal L)-1
\]
and
\[
\bigwedge S_0=\bigwedge S.
\]
Here, as usual,
\[
\bigwedge\varnothing=1.
\]
	\begin{proof}
		Let
		\[
		m=\bigwedge S.
		\]
		If $S=\varnothing$, take $S_0=\varnothing$. Then
		\[
		\bigwedge S_0
		=
		\bigwedge\varnothing
		=
		1
		=
		\bigwedge S,
		\]
		and
		\[
		|S_0|=0\leq h(\mathcal L)-1.
		\]
		Now suppose that $S\neq\varnothing$. Starting with $c_0=1$, construct a descending sequence as follows.\\
		If $c_0=m$, then we may again take $S_0=\varnothing$, and the result
		follows. Otherwise,
		\[
		m<c_0.
		\]
		More generally, suppose that $c_i\neq m$. Since
		\[
		m=\bigwedge S\leq c_i,
		\]
		we have
		\[
		m<c_i.
		\]
		We claim that there exists $s_i\in S$ such that
		\[
		c_i\wedge s_i<c_i.
		\]
		Indeed, if
		\[
		c_i\wedge s=c_i
		\]
		for every $s\in S$, then
		\[
		c_i\leq s
		\]
		for every $s\in S$. Hence $c_i$ is a lower bound of $S$, and therefore
		\[
		c_i\leq\bigwedge S=m,
		\]
		contradicting $m<c_i$. Choose such an $s_i\in S$ and define
		\[
		c_{i+1}=c_i\wedge s_i.
		\]
		Then
		\[
		c_{i+1}<c_i.
		\]
		Repeating this construction, we obtain a strictly descending chain
		\[
		1=c_0>c_1>\cdots>c_r.
		\]
		Since $\mathcal L$ is finite, this process cannot continue indefinitely.
		It must terminate at some $c_r$. Moreover, it can terminate only when
		\[
		c_r=m;
		\]
		otherwise, if $c_r\neq m$, the preceding argument would produce
		$s_r\in S$ such that
		\[
		c_r\wedge s_r<c_r,
		\]
		contradicting termination. Hence
		\[
		1=c_0>c_1>\cdots>c_r=m.
		\]
		By definition, $h(\mathcal L)$ is the maximum number of elements in a
		chain of $\mathcal L$. The above chain contains $r+1$ elements, so
		\[
		r+1\leq h(\mathcal L),
		\]
		and therefore
		\[
		r\leq h(\mathcal L)-1.
		\]
		
		Now put
		\[
		S_0=\{s_0,\ldots,s_{r-1}\}.
		\]
		Then
		\[
		|S_0|\leq r\leq h(\mathcal L)-1.
		\]
		Also, by construction,
		\[
		\begin{aligned}
			c_r
			&=
			c_{r-1}\wedge s_{r-1}\\
			&=
			s_0\wedge s_1\wedge\cdots\wedge s_{r-1}\\
			&=
			\bigwedge S_0.
		\end{aligned}
		\]
		Since $c_r=m=\bigwedge S$, we conclude that
		\[
		\bigwedge S_0=\bigwedge S.
		\]
		Thus
		\[
		|S_0|\leq h(\mathcal L)-1
		\qquad\text{and}\qquad
		\bigwedge S_0=\bigwedge S.
		\]

		
	\end{proof}
\end{lem}
The preceding lemma shows that the exact value of a formula of the
form $\Box\psi$ at a state can be determined by a bounded number of
its successors. More precisely, at most $h(\mathcal L)-1$ suitably chosen successors
are sufficient to preserve the meet defining the value of the
$\Box$-formula.

\begin{cor}\label{cor:modal-witness-set}
	Let $M=(W,\mathcal R,V)$ be an $\mathcal L$-valued Kripke model,
	let $w\in W$, and let $\psi$ be a formula. There exists a set
	\[
	E(w,\psi)\subseteq\mathcal R[w]
	\]
	such that
	\[
	|E(w,\psi)|\leq h(\mathcal L)-1
	\]
	and
	\[
	\bigwedge_{y\in E(w,\psi)}
	\llbracket\psi\rrbracket_M(y)
	=
	\llbracket\Box\psi\rrbracket_M(w).
	\]
	\begin{proof}
	Let
	\[
	S=
	\left\{
	\llbracket\psi\rrbracket_M(y):
	y\in\mathcal R[w]
	\right\}
	\subseteq\mathcal L.
	\]
	Then, by the semantics of $\Box$,
	\[
	\llbracket\Box\psi\rrbracket_M(w)
	=
	\bigwedge S.
	\]
	By Lemma~\ref{lem:meet-compression}, there exists
	$S_0\subseteq S$ such that
	\[
	|S_0|\leq h(\mathcal L)-1
	\]
	and
	\[
	\bigwedge S_0=\bigwedge S.
	\]
	For each $\ell\in S_0$, choose a successor
	$y_\ell\in\mathcal R[w]$ such that
	\[
	\llbracket\psi\rrbracket_M(y_\ell)=\ell,
	\]
	and define
	\[
	E(w,\psi)=\{y_\ell:\ell\in S_0\}.
	\]
	Then
	\[
	|E(w,\psi)|
	\leq |S_0|
	\leq h(\mathcal L)-1.
	\]
	Moreover,
	\[
	\begin{aligned}
		\bigwedge_{y\in E(w,\psi)}
		\llbracket\psi\rrbracket_M(y)
		&=
		\bigwedge_{\ell\in S_0}\ell\\
		&=
		\bigwedge S_0\\
		&=
		\bigwedge S\\
		&=
		\llbracket\Box\psi\rrbracket_M(w).
	\end{aligned}
	\]
	If $\mathcal R[w]=\varnothing$, take
	$E(w,\psi)=\varnothing$; then both meets are equal to $1$.
	\end{proof}
\end{cor}

\subsection{The bounded certificate construction}

For a formula $\varphi$, let
\[
\operatorname{Sub}_{\Box}(\varphi)
=
\{
\Box\psi:
\Box\psi\text{ is a subformula of }\varphi
\},
\]
and put
\[
m_{\Box}(\varphi)
=
\bigl|
\operatorname{Sub}_{\Box}(\varphi)
\bigr|.
\]
Let
\[
d=\operatorname{md}(\varphi)
\]
be the modal depth of $\varphi$, and define
\[
b_{\mathcal L,\varphi}
=
\bigl(h(\mathcal L)-1\bigr)m_{\Box}(\varphi).
\]
Fix an $\mathcal L$-valued Kripke model
\[
M=(W,\mathcal R,V)
\]
and a state $x\in W$. We construct a rooted tree-like model
\[
C(M,x,\varphi)=(W_C,\mathcal R_C,V_C),
\]
which will serve as a finite certificate for the value of $\varphi$
at $x$.

The nodes of $W_C$ are copies of states of the original model.
To distinguish a node of the certificate from the state of $M$ that
it represents, we use a labelling map
\[
\lambda:W_C\longrightarrow W.
\]
The root of the tree is denoted by $\varepsilon$ and is labelled by
the distinguished state $x$:
\[
\lambda(\varepsilon)=x.
\]
If $u\in W_C$ is a node and $y\in W$ is selected as a successor of
$\lambda(u)$, then
\[
u^\frown y
\]
denotes a new child of $u$, with
\[
\lambda(u^\frown y)=y.
\]
Thus different nodes of $W_C$ may carry the same label from $W$.
In particular, the construction should be regarded as a tree of
labelled copies of states of $M$, rather than as a submodel of $M$.

Suppose now that $u\in W_C$ has depth $n<d$ and
\[
\lambda(u)=w.
\]
For every subformula
\[
\Box\psi\in\operatorname{Sub}_{\Box}(\varphi)
\]
such that
\[
\operatorname{md}(\Box\psi)\leq d-n,
\]
Corollary~\ref{cor:modal-witness-set} yields a set
\[
E_{u,\Box\psi}\subseteq\mathcal R[w]
\]
satisfying
\[
|E_{u,\Box\psi}|\leq h(\mathcal L)-1
\]
and
\[
\bigwedge_{y\in E_{u,\Box\psi}}
\llbracket\psi\rrbracket_M(y)
=
\llbracket\Box\psi\rrbracket_M(w).
\]
The set $E_{u,\Box\psi}$ contains the successors needed to preserve
the exact value of the particular subformula $\Box\psi$ at the state
$w$.

We collect the witnesses required for all relevant $\Box$-subformulas
by putting
\[
E_u=
\bigcup
\left\{
E_{u,\Box\psi}:
\Box\psi\in\operatorname{Sub}_{\Box}(\varphi),
\ \operatorname{md}(\Box\psi)\leq d-n
\right\}.
\]
For every $y\in E_u$, introduce one child $u^\frown y$ of $u$ and set
\[
\lambda(u^\frown y)=y.
\]
The relation $\mathcal R_C$ is the parent--child relation generated by
these choices:
\[
u\,\mathcal R_C\,u^\frown y
\qquad\text{whenever }y\in E_u.
\]
No children are introduced at nodes of depth $d$.

The following diagram illustrates one step of the construction. If
$y_1,y_2\in E_u$, then the certificate contains children labelled by
the corresponding successors of $\lambda(u)$ in the original model:
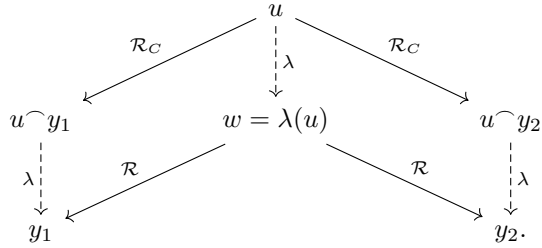
\[
\begin{tikzcd}[row sep=large,column sep=huge]
	&
	u
	\arrow[dl,"\mathcal R_C"']
	\arrow[dr,"\mathcal R_C"]
	\arrow[d,dashed,"\lambda"]
	&
	\\
	u^\frown y_1
	\arrow[d,dashed,"\lambda"']
	&
	w=\lambda(u)
	\arrow[dl,"\mathcal R"']
	\arrow[dr,"\mathcal R"]
	&
	u^\frown y_2
	\arrow[d,dashed,"\lambda"]
	\\
	y_1
	&&
	y_2 .
\end{tikzcd}
\]
Thus every edge of the certificate follows an edge of the original
model under the labelling map, but only the successors required for
the relevant $\Box$-subformulas are retained.

Finally, the valuation on the certificate is inherited through
$\lambda$. For every propositional variable $p$ occurring in
$\varphi$, define
\[
V_C(p)(u)=V(p)(\lambda(u)).
\]
Hence a node of the certificate and the original state labelling it
have exactly the same atomic truth values.

The construction has depth at most
\[
d=\operatorname{md}(\varphi),
\]
and each node has only finitely many children. The precise branching
and size bounds are established below. The next theorem shows that
this pruning of the original successor structure does not change the
exact value of any subformula that is relevant at the corresponding
depth.






\begin{thm}[Bounded exact-value certificate]
	\label{thm:bounded-certificate}
	Let $u\in W_{\mathcal C}$ have depth $n$. For every subformula
	$\theta$ of $\varphi$ satisfying
	\[
	\operatorname{md}(\theta)\leq d-n,
	\]
	we have
	\[
	\llbracket\theta\rrbracket_{\mathcal C}(u)
	=
	\llbracket\theta\rrbracket_M(\lambda(u)).
	\]
	In particular,
	\[
	\llbracket\varphi\rrbracket_{\mathcal C}(\varepsilon)
	=
	\llbracket\varphi\rrbracket_M(x).
	\]
	\begin{proof}
		We prove the statement by structural induction on $\theta$,
		simultaneously for all nodes $u\in W_C$ satisfying the indicated
		depth condition.
		
		If $\theta=p$ is a propositional variable occurring in $\varphi$,
		then, by the definition of $V_C$,
		\[
		\llbracket p\rrbracket_C(u)
		=
		V_C(p)(u)
		=
		V(p)(\lambda(u))
		=
		\llbracket p\rrbracket_M(\lambda(u)).
		\]
		The cases $\theta=0$ and $\theta=1$ are immediate. The cases of $\wedge$, $\vee$, $\to$, and the truth-value tests
		$T_\ell$ follow directly from the induction hypothesis, since these
		operations are evaluated pointwise.
		
		Suppose now that
		\[
		\theta=\Box\psi
		\]
		and that $u$ has depth $n$. By assumption,
		\[
		\operatorname{md}(\Box\psi)\leq d-n.
		\]
		Hence $\Box\psi$ is one of the relevant boxed subformulas considered
		in the construction at $u$, and therefore the witness set
		\[
		E_{u,\Box\psi}\subseteq E_u
		\]
		was selected.
		
		Every $y\in E_u$ gives rise to a child $u^\frown y$ of $u$, and
		\[
		\lambda(u^\frown y)=y.
		\]
		Since
		\[
		\operatorname{md}(\Box\psi)
		=
		\operatorname{md}(\psi)+1
		\leq d-n,
		\]
		we have
		\[
		\operatorname{md}(\psi)
		\leq d-(n+1).
		\]
		The child $u^\frown y$ has depth $n+1$. Hence the induction
		hypothesis applies to $\psi$ at every such child, giving
		\[
		\llbracket\psi\rrbracket_C(u^\frown y)
		=
		\llbracket\psi\rrbracket_M(y)
		\qquad (y\in E_u).
		\]
		Therefore,
		\[
		\begin{aligned}
			\llbracket\Box\psi\rrbracket_C(u)
			&=
			\bigwedge_{y\in E_u}
			\llbracket\psi\rrbracket_C(u^\frown y)\\
			&=
			\bigwedge_{y\in E_u}
			\llbracket\psi\rrbracket_M(y).
		\end{aligned}
		\]
		Now
		\[
		E_{u,\Box\psi}
		\subseteq
		E_u
		\subseteq
		\mathcal R[\lambda(u)].
		\]
		Since taking the meet over a larger set can only decrease the meet,
		\[
		\begin{aligned}
			\bigwedge_{y\in\mathcal R[\lambda(u)]}
			\llbracket\psi\rrbracket_M(y)
			&\leq
			\bigwedge_{y\in E_u}
			\llbracket\psi\rrbracket_M(y)\\
			&\leq
			\bigwedge_{y\in E_{u,\Box\psi}}
			\llbracket\psi\rrbracket_M(y).
		\end{aligned}
		\]
		By the choice of $E_{u,\Box\psi}$,
		\[
		\bigwedge_{y\in E_{u,\Box\psi}}
		\llbracket\psi\rrbracket_M(y)
		=
		\llbracket\Box\psi\rrbracket_M(\lambda(u))
		\]
		and, by the semantics of $\Box$,
		\[
		\llbracket\Box\psi\rrbracket_M(\lambda(u))
		=
		\bigwedge_{y\in\mathcal R[\lambda(u)]}
		\llbracket\psi\rrbracket_M(y).
		\]
		Thus the first and last terms in the preceding inequalities are equal.
		Hence
		\[
		\bigwedge_{y\in E_u}
		\llbracket\psi\rrbracket_M(y)
		=
		\llbracket\Box\psi\rrbracket_M(\lambda(u)).
		\]
		Consequently,
		\[
		\llbracket\Box\psi\rrbracket_C(u)
		=
		\llbracket\Box\psi\rrbracket_M(\lambda(u)).
		\]
		
		This completes the induction.
		
		Finally, the root $\varepsilon$ has depth $0$,
		$\lambda(\varepsilon)=x$, and
		\[
		\operatorname{md}(\varphi)=d.
		\]
		Applying the result with $\theta=\varphi$ gives
		\[
		\llbracket\varphi\rrbracket_C(\varepsilon)
		=
		\llbracket\varphi\rrbracket_M(x).
		\]

		
		
	\end{proof}
	
\end{thm}

\begin{prop}[Size bound]\label{prop:certificate-size}
Every node of $C(M,x,\varphi)$ has at most
$b_{\mathcal L,\varphi}$ children. Consequently,
\[
|W_C|
\leq
1+b_{\mathcal L,\varphi}
+b_{\mathcal L,\varphi}^{2}
+\cdots+
b_{\mathcal L,\varphi}^{d}.
\]
In particular,
\[
|W_C|\leq
\begin{cases}
	1,
	& b_{\mathcal L,\varphi}=0,\\[1mm]
	d+1,
	& b_{\mathcal L,\varphi}=1,\\[2mm]
	\displaystyle
	\frac{
		b_{\mathcal L,\varphi}^{\,d+1}-1
	}{
		b_{\mathcal L,\varphi}-1
	},
	& b_{\mathcal L,\varphi}>1.
\end{cases}
\]
	\begin{proof}
	Let $u\in W_C$. For each relevant subformula
	$\Box\psi\in\operatorname{Sub}_{\Box}(\varphi)$, the construction
	selects a set
	\[
	E_{u,\Box\psi}
	\]
	with
	\[
	|E_{u,\Box\psi}|
	\leq h(\mathcal L)-1.
	\]
	Since
	\[
	E_u
	=
	\bigcup
	\left\{
	E_{u,\Box\psi}:
	\Box\psi\in\operatorname{Sub}_{\Box}(\varphi),
	\ \operatorname{md}(\Box\psi)\leq d-\operatorname{depth}(u)
	\right\},
	\]
	we obtain
	\[
	\begin{aligned}
		|E_u|
		&\leq
		\sum_{\Box\psi\in\operatorname{Sub}_{\Box}(\varphi)}
		|E_{u,\Box\psi}|\\
		&\leq
		\bigl(h(\mathcal L)-1\bigr)m_{\Box}(\varphi)\\
		&=
		b_{\mathcal L,\varphi}.
	\end{aligned}
	\]
	Thus every node has at most $b_{\mathcal L,\varphi}$ children. The construction has depth at most $d$. Hence the number of nodes at depth $n$ is at most $b_{\mathcal L,\varphi}^{\,n}$, and therefore
	\[
	|W_C|
	\leq
	\sum_{n=0}^{d} b_{\mathcal L,\varphi}^{\,n}.
	\]
	The three displayed bounds follow from this finite geometric sum.
	\end{proof}
	
\end{prop}

The bound in Proposition \ref{prop:certificate-size} depends only on
$\mathcal L$ and the syntactic structure of $\varphi$. In particular,
it is independent of the number of states of the original model.

\subsection{Reduced counterexample extraction}\label{subsec:counterexample-extraction}

Let
\[
\Sigma_\varphi=\operatorname{Var}(\varphi)
\]
and let
\[
C=C(M,x,\varphi)
=(W_C,\mathcal R_C,V_C)
\]
be the bounded certificate constructed above. By
Proposition~\ref{prop:certificate-size}, the model $C$ is finite.
Hence the reduction of Section~4 can be applied to $C$ relative to
the vocabulary $\Sigma_\varphi$. Let
\[
\equiv_C^{\Sigma_\varphi}
\]
denote the corresponding observational equivalence on $W_C$, and
write
\[
C_\varphi^{\mathrm{red}}
\]
for the resulting reduced model. Its state set is
\[
W_C/{\equiv_C^{\Sigma_\varphi}}.
\]
In particular, $[\varepsilon]$ denotes the equivalence class of the
root $\varepsilon$ of the certificate.

\begin{thm}[Reduced exact-value certificate]
	\label{thm:reduced-certificate}
Let
\[
a=\llbracket\varphi\rrbracket_M(x).
\]
Then
\[
\llbracket\varphi\rrbracket_
{C_\varphi^{\mathrm{red}}}([\varepsilon])
=
a.
\]
Moreover,
\[
|W_{C_\varphi^{\mathrm{red}}}|
\leq
1+b_{\mathcal L,\varphi}
+b_{\mathcal L,\varphi}^{2}
+\cdots+
b_{\mathcal L,\varphi}^{d}.
\]
	\begin{proof}
		By Theorem~\ref{thm:bounded-certificate},
		\[
		\llbracket\varphi\rrbracket_C(\varepsilon)
		=
		\llbracket\varphi\rrbracket_M(x)
		=
		a.
		\]
		Since $\Sigma_\varphi$ contains all propositional variables occurring
		in $\varphi$, Theorem \ref{thm:exact-preservation} applies to the reduction of $C$. Hence
		\[
		\llbracket\varphi\rrbracket_
		{C_\varphi^{\mathrm{red}}}([\varepsilon])
		=
		\llbracket\varphi\rrbracket_C(\varepsilon)
		=
		a.
		\]
		
		Furthermore, $C_\varphi^{\mathrm{red}}$ is a quotient of $C$, so
		\[
		|W_{C_\varphi^{\mathrm{red}}}|
		\leq |W_C|.
		\]
		The required bound therefore follows from
		Proposition~\ref{prop:certificate-size}.

	\end{proof}
\end{thm}
Thus the two constructions preserve the same exact value:
\[
\llbracket\varphi\rrbracket_M(x)
=
\llbracket\varphi\rrbracket_C(\varepsilon)
=
\llbracket\varphi\rrbracket_
{C_\varphi^{\mathrm{red}}}([\varepsilon]).
\]

\begin{cor}[Counterexample extraction]
	\label{cor:counterexample-extraction}
If
\[
\llbracket\varphi\rrbracket_M(x)\neq1,
\]
then
\[
\bigl(C_\varphi^{\mathrm{red}},[\varepsilon]\bigr)
\]
is a finite pointed counterexample to $\varphi$. Moreover, it
preserves the exact failure value:
\[
\llbracket\varphi\rrbracket_
{C_\varphi^{\mathrm{red}}}([\varepsilon])
=
\llbracket\varphi\rrbracket_M(x).
\]
\begin{proof}
	This follows immediately from Theorem~\ref{thm:reduced-certificate}.
\end{proof}
\end{cor}

Thus a failure of $\varphi$ at a state $x$ is handled in two stages.
First, the construction of $C(M,x,\varphi)$ retains only bounded
successor information needed to preserve the exact value of
$\varphi$ at $x$. Second, the observational quotient may further
identify certificate states that agree on all modal observations
generated by $\Sigma_\varphi$. Hence the first stage is
formula-directed, whereas the second is observation-directed.
\begin{rem}\label{rem:not-minimal-countermodel}
The model $C_\varphi^{\mathrm{red}}$ is not claimed to be a smallest
countermodel to $\varphi$ among all $\mathcal L$-valued models. Theorem \ref{thm:minimality}, the
minimality result of Section~4, applies only to surjective quotients
of $C$ preserving all observations generated by
$\Sigma_\varphi$.
\end{rem}
\section{Conclusion and future work}\label{6}
We have developed a finite reduction method for Fitting's
$\mathcal L$-valued modal logic using its relational bitopological
representation. For a finite $\mathcal L$-valued Kripke model
$M=(W,\mathcal R,V)$ and a finite vocabulary $\Sigma$, the
state-evaluation map
\[
\eta_M^\Sigma:
W\longrightarrow
\operatorname{Hom}_{\mathcal{VA}_{\mathcal L}}
(A_M^\Sigma,\mathcal L)
\]
identifies precisely those states that agree on all observations
generated by $\Sigma$. The resulting observational quotient
$W_M^\Sigma$ is isomorphic, as a finite relational bitopological
structure, to the image $X_M^\Sigma=\eta_M^\Sigma[W]$ in the dual
space. Under this identification, the reduced accessibility relation
$\mathcal R^{\mathrm{red}}$ coincides with the restriction of the
canonical dual relation. Consequently, every formula over $\Sigma$
preserves its exact $\mathcal L$-truth value under the reduction.
The universal factorization property further shows that this
minimality is relative to surjective reductions preserving the
generated $\Sigma$-observations.

We have also obtained bounded exact-value certificates for individual
formula evaluations. The meet-compression argument shows that the
value of a formula of the form $\Box\psi$ can be witnessed by at most
$h(\mathcal L)-1$ suitably chosen successors. Iterating this
construction up to the modal depth of a formula $\varphi$ yields a
finite tree-like certificate whose branching and size are bounded in
terms of $\mathcal L$ and the syntactic structure of $\varphi$, and
are independent of the size of the original model. Applying the
observational quotient to this certificate gives a finite reduced
pointed counterexample whenever $\varphi$ fails, while preserving its
exact failure value.\\
Several directions remain open. From a computational point of view,
it would be useful to obtain the observational quotient without
explicitly generating the whole algebra $A_M^\Sigma$, for example by
developing a direct refinement procedure based on the observations
relevant to $\Sigma$.

A more substantial direction concerns geometric and fuzzy geometric
logic. Fuzzy geometric logic and its graded consequence relation were
introduced by Chakraborty and Jana
\cite{chakraborty2017fuzzy}, together with the associated notions of
graded frames and fuzzy topological systems. Coalgebraic geometric
logic was subsequently developed by Bezhanishvili et al. \cite{bezhanishvili2022coalgebraic}, using coalgebras over
topological spaces and open predicate liftings. A fuzzy coalgebraic
counterpart was developed by Das et al. \cite{das2024coalgebraic}, where modalities are described by fuzzy-open
predicate liftings over fuzzy topological spaces.

It is an open problem whether suitable finitary fragments of these
fuzzy coalgebraic models admit formula-dependent finite reductions
or bounded exact-value certificates analogous to those obtained
here. Such an extension is not immediate: the meet-compression
argument of the present paper depends essentially on the finite
height of $\mathcal L$, whereas a general fuzzy truth-value algebra
may have infinite height and an infimum need not be determined by
finitely many values. Moreover, any quotient construction would have
to preserve the relevant fuzzy-open predicate liftings and the
underlying coalgebraic structure. Identifying conditions under which
finite exact-value reduction remains possible therefore constitutes
a natural problem for further investigation.

\section*{Funding}
This research received no external funding.
\section*{Conflict of Interest Statement}
The authors declare that they have no known competing financial interests or personal relationships that could have appeared to influence the work reported in this paper.
\section*{Data Availability Statement}
Data sharing is not applicable to this article as no new data were created or analyzed in this study.

\end{document}